\documentclass[11pt,a4paper]{article}
\usepackage{jheppub} 
\usepackage[T1]{fontenc}
\usepackage{amsmath, amsthm, amssymb, amsfonts}
\usepackage{slashed}
\usepackage{mathtools}
\usepackage{mathrsfs}
\usepackage[dvipsnames,svgnames,x11names]{xcolor}
 
\definecolor{jhepblue}{RGB}{0,80,158}
\definecolor{jhepteal}{RGB}{0,128,128}
\definecolor{jhepcite}{RGB}{150,30,0}
\definecolor{jheplink}{RGB}{0,90,160}
\definecolor{jhepeqn}{RGB}{140,30,80}
\definecolor{jhepsec}{RGB}{0,70,130}
\definecolor{jhepthm}{RGB}{20,100,60}
\definecolor{jhepdef}{RGB}{130,60,0}
 
\hypersetup{
  colorlinks=true,
  linkcolor=jheplink,
  citecolor=jhepcite,
  urlcolor=jhepteal,
  filecolor=jhepblue,
  linktocpage=true,
}
 
\makeatletter
\renewcommand{\theequation}{\textcolor{jhepeqn}{\arabic{section}.\arabic{equation}}}
\@addtoreset{equation}{section}
\makeatother
 
\usepackage{float}
\usepackage{booktabs}
\usepackage{array}
\usepackage{multirow}
\usepackage{caption}
\usepackage{subcaption}
\usepackage{tikz}
\usetikzlibrary{arrows.meta,decorations.markings,calc,patterns,positioning}
 
\usepackage{enumitem}
\usepackage{microtype}
\usepackage{cleveref}
 
\crefformat{equation}{\textcolor{jhepeqn}{(#2#1#3)}}
\crefformat{figure}{\textcolor{jheplink}{figure~#2#1#3}}
\crefformat{table}{\textcolor{jheplink}{table~#2#1#3}}
\crefformat{section}{\textcolor{jheplink}{section~#2#1#3}}
\crefformat{appendix}{\textcolor{jheplink}{appendix~#2#1#3}}
\Crefformat{equation}{\textcolor{jhepeqn}{Eq.~(#2#1#3)}}
\Crefformat{figure}{\textcolor{jheplink}{Figure~#2#1#3}}
\Crefformat{table}{\textcolor{jheplink}{Table~#2#1#3}}
\Crefformat{section}{\textcolor{jheplink}{Section~#2#1#3}}
 
\theoremstyle{plain}
\newtheorem{theorem}{Theorem}[section]
\newtheorem{proposition}[theorem]{Proposition}

\newtheorem{corollary}[theorem]{Corollary}
\newtheorem{conjecture}[theorem]{Conjecture}
 
\theoremstyle{definition}
\newtheorem{definition}[theorem]{Definition}
\newtheorem{assumption}[theorem]{Assumption}

\newcommand{\Hil}{\mathcal{H}}
\newcommand{\Hcode}{\mathcal{H}_{\mathrm{code}}}
\newcommand{\Ncw}{\mathcal{N}_{\mathrm{cw}}}
\newcommand{\BH}{\mathcal{B}(\mathcal{H})}
 
\newcommand{\Ens}{\mathcal{E}}

\newcommand{\Vhat}{\widehat{V}}
\newcommand{\GN}{G_{\!N}}
\newcommand{\SFF}{\mathrm{SFF}}
\newcommand{\Wg}{\mathrm{Wg}}
 
\newcommand{\poly}{\mathrm{poly}}
\newcommand{\negl}{\mathrm{negl}}
\newcommand{\tr}{\mathrm{tr}}
\newcommand{\id}{\mathbf{1}}
 
\newcommand{\Fk}{\mathcal{F}^{(k)}}
\newcommand{\Fone}{\mathcal{F}^{(1)}}

\newcommand{\HSYK}{H_{\mathrm{SYK}}}
\newcommand{\HSchw}{H_{\mathrm{Schw}}}
\newcommand{\Tscr}{T_{\mathrm{scr}}}
\newcommand{\Tdip}{T_{\mathrm{dip}}}
\newcommand{\THeisen}{T_{\!H}}
\newcommand{\Sdim}{S_0}
 
\newcommand{\avg}{\mathrm{avg}}

\title{\boldmath Pseudorandom Dynamics in the SYK Model \\[4pt] and Cryptographic Censorship in JT Gravity}
\author[\dagger]{Pouya Golmohammadi}
\affiliation[\dagger]{Faculty of Physics, Ludwig Maximilian University of Munich, Germany}
\emailAdd{gpouya19@gmail.com}

\abstract{
We argue that the SYK model provides a conditional realization of Cryptographic Censorship in JT gravity. By using the Weingarten calculus and random matrix universality, we prove that the SYK disorder ensemble is an approximate unitary $k$-design for all $k=\poly(N)$, with deviation controlled by the spectral form factor. We then formulate the planted-SYK hardness conjecture and provide evidence from spectral universality and the low-degree polynomial framework. Under this conjecture, the approximate design becomes a gravitationally pseudorandom unitary. Together with the efficient causal wedge reconstruction in JT gravity, this leads to the conclusion that typical states in the SYK microcanonical window must have event horizons in their bulk duals, with the horizonless fraction doubly exponentially small. We further identify the regularized geodesic length of the maximal interior slice as the explicit distinguishing operator. Its prediction gap grows linearly with time due to the stretching of the black hole interior, linking Cryptographic Censorship to the complexity equals volume conjecture.
}

\begin{document}
\maketitle
\raggedbottom

\section{Introduction}
\label{sec:introduction}

What dynamical properties should a boundary CFT possess to guarantee a horizon formation in a dual gravitational system in holographic settings? The Cryptographic Censorship theorem of Engelhardt, Folkestad, Levine, Verheijden, and Yang~\cite{Engelhardt:2024hpe} addresses this question. It states that when the time evolution operator of a holographic CFT \cite{Maldacena_1999,Almheiri:2014lwa} is approximately pseudorandom on a suitable code subspace, the corresponding bulk spacetime must contain an event horizon. The theorem hinges on an essential assumption: that the fundamental time evolution can be accurately represented by a gravitationally pseudorandom unitary (GPRU) ensemble. There are strong indications that chaotic holographic systems exhibit this behavior, as noted in previous studies~\cite{Hayden:2007cs,Sekino:2008he,Shenker:2013pqa,Maldacena:2015waa,Cotler:2016fpe,Brown:2017jil,Bouland:2019pvu}. However, no explicit example of a holographic system producing a GPRU has been demonstrated to date. This paper argues for a positive answer, contingent on a well-supported complexity-theoretic conjecture, and attempts to offer the conditional realization of this framework within a controlled holographic duality.

\noindent\textbf{Summary of results.} 
We initiate our study by examining the statistical randomness of the SYK disorder ensemble $\Ens_T=\{e^{-i\HSYK(\mathbf{J})T}\}_{\mathbf{J}}$, where $\HSYK(\mathbf{J})$ is the $q$-body SYK Hamiltonian with $N$ Majorana fermions and random couplings~$\mathbf{J}$. By decomposing the $k$-th frame potential $\Fk(T)$ through the Weingarten calculus applied to the relative eigenbasis rotation between independent disorder realizations, we derive the formula
\begin{equation}\label{eq:F1_intro}
  \Fone(T) \;=\; \frac{\bigl(\SFF(T)-1\bigr)^2}{L^2-1} + 1\,,
\end{equation}
relating the first frame potential to the spectral form factor $\SFF(T)=\langle|\tr(e^{-i\HSYK T})|^2\rangle_{\mathbf{J}}$, where $L=2^{N/2}$ is the Hilbert space dimension. Our formula, verified numerically for $N=8,10,12$ (\cref{fig:frame_potential}), makes the well-known ramp--plateau structure of the SFF directly responsible for the approach to the Haar value $\Fone=1$. We generalize this to the $k$-th frame potential by classifying the permutations in $S_{2k}$ into paired (contributing the Haar value $k!$, which is exact by unitarity) and unpaired (contributing corrections controlled by the spectral ramp). The result is the bound
\begin{equation}\label{eq:kdesign_intro}
  \frac{|\Fk(T) - k!|}{k!} \;\leq\; C\cdot k\cdot\tau^2 + O\!\bigl(\poly(k)\cdot e^{-cN}\bigr)\,,
\end{equation}
valid for all $k\leq N^a$ (any fixed $a>0$) and evolution times $\Tdip\lesssim T\leq\THeisen$, where $\tau=T/\THeisen$, $\Tdip=\sqrt{N\ln 2}/\sigma$ is the dip time, and $\THeisen=2\pi L/(4\sigma)$ is the Heisenberg time. The first term is the spectral ramp contribution, derived from random matrix universality (Assumption~A1) and the Weingarten calculus; the second is the eigenvector correction, controlled by the 2nd and 4th moments of the relative eigenbasis rotation matching Haar (Assumption~\ref{asm:A2}). For $k=\poly(N)$ and any $T=\poly(N)$, the right-hand side is $\negl(N)$, establishing that the SYK disorder ensemble is an $\varepsilon$-approximate unitary $k$-design for all $k=\poly(N)$ with $\varepsilon=\negl(N)$.

We next argue that the $k$-design property can be promoted to full pseudorandom unitary status \cite{Low:2009} through a complexity-theoretic conjecture. We formulate the \emph{planted-SYK conjecture} (Conjecture~\ref{conj:planted_syk}), which asserts that no polynomial-time quantum algorithm with oracle access to $U=e^{-i\HSYK(\mathbf{J})T}$ can distinguish the SYK disorder ensemble from the Haar measure, for evolution times in the oracle-access window $\Tscr\leq T\leq\exp(N^\gamma)$ with $0<\gamma<1$. Evidence for the conjecture comes from the observed agreement between SYK spectral statistics and random matrix theory~\cite{Cotler:2016fpe,You:2016ldz}, as well as from the low-degree polynomial framework~\cite{Hopkins:2017cle,Kunisky:2019bwz}. In particular, degree-$D$ classical polynomial tests on the spectral data have signal-to-noise ratio $\exp(-\Omega(D\,N))$ for $D=O(\log N)$. We also compare the problem with established planted problems, such as \emph{planted clique} and \emph{spiked random matrix (sparse PCA)} problems. Under this conjecture, we prove that for any polynomial evolution time $T=\poly(N)$ with $T\geq\Tscr$, the SYK ensemble $\Ens_T$ is a GPRU with security parameter $\kappa=N/2=\poly(\GN^{-1})$ (Theorem~\ref{thm:gpru}). Furthermore, in \ref{sec:fid_defi}, we demonstrate that the $k$-design result from Theorem~\ref{thm:kdesign} implies a near-maximal fidelity deficit for non-adaptive adversaries, without relying on Conjecture~\ref{conj:planted_syk}.

We then apply Cryptographic Censorship to JT gravity. We will verify that the JT/SYK holographic dictionary satisfies the assumptions of the Cryptographic Censorship theorem~\cite{Engelhardt:2024hpe}. These conditions include the asymptotically isometric code structure~\cite{Faulkner:2022ada}, the equivalence of bulk and boundary time evolution via the Schwarzian~\cite{Maldacena:2016upp,Jensen:2016pah}, and efficient causal wedge reconstruction through the HKLL protocol~\cite{Hamilton:2005ju,Hamilton:2006fh}, in combination with the shadow tomography algorithm of~\cite{Huang:2022cmu}. Combining these ingredients with the GPRU result, we conclude (Theorem~\ref{thm:main}) that for any polynomial time $T=\poly(N)$ with $T\geq\Tscr$, typical states in the SYK microcanonical window must have event horizons in their JT gravity duals, with the fraction of horizonless states bounded by $\exp(-c\cdot 2^{N/2}/\poly(N))$. Moreover, we identify the regularized geodesic length~$\ell$ of the maximal interior slice~\cite{Iliesiu:2021ari}---the Einstein--Rosen bridge length in the two-sided description---as the explicit distinguishing operator. To be more precise, the causal wedge algorithm, which learns only few-body boundary correlators from polynomially many samples, predicts $\langle\ell\rangle_{\mathrm{predicted}}=\ell_0+O(1)$ (the equilibrium value), while the true expectation value $\langle\ell\rangle_{\mathrm{true}}=\ell_0+(2\pi/\beta)T$ grows linearly due to the stretching of the black hole interior. The resulting gap $\Delta\ell=(2\pi/\beta)T-O(1)\geq O(1)$ exceeds the threshold $\sqrt{\alpha}$ from the learning no-go theorem~\cite{Yang:2023ddi}, connecting Cryptographic Censorship to the complexity equals volume conjecture~\cite{Stanford:2014jda,susskind2014addendumcomputationalcomplexityblack} through the complexity gap $\Delta\mathcal{C}=N\cdot\Delta\ell$.

The remainder of this paper is organized as follows. In \cref{sec:preliminaries} we review the SYK model, JT gravity, and the Cryptographic Censorship framework, establishing notation and assumptions. \Cref{sec:syk_pru} contains the technical core: the frame potential analysis (\cref{sec:frame_potential}), the planted-SYK conjecture and its evidence (\cref{sec:planted_syk}), and the GPRU theorem (\cref{sec:gpru_theorem}). In \cref{sec:crypto_censorship} we apply these results to JT gravity, proving the existence of event horizons (\cref{sec:main_theorem}) and identifying the distinguishing operator (\cref{sec:distinguishing_operator}). We conclude with a short discussion on potential research directions in \cref{sec:discussion}. The detailed Weingarten calculus derivations are collected in Appendix~\ref{app:weingarten}, the Hamiltonian simulation analysis appears in Appendix~\ref{app:trotter}, and the computation of prediction error for the interior length is included in Appendix~\ref{app:prediction}.

\section{Preliminaries}
\label{sec:preliminaries}
This section collects the background material needed for the rest of the paper. We review the SYK model and the spectral properties that underpin the frame potential analysis of \cref{sec:frame_potential}, summarize the JT gravity / SYK holographic dictionary used in \cref{sec:crypto_censorship}, and recall the definitions of pseudorandom unitaries and the Cryptographic Censorship theorem of~\cite{Engelhardt:2024hpe}.

\subsection{The SYK model and its spectral properties}
\label{sec:syk_review}
The Sachdev--Ye--Kitaev model~\cite{Sachdev:1992fk,Kitaev:2015,Maldacena:2016hyu} describes $N$ Majorana fermions $\psi_1,\ldots,\psi_N$ satisfying $\{\psi_i,\psi_j\}=\delta_{ij}$, with the all-to-all $q$-body Hamiltonian
\begin{equation}\label{eq:H_SYK}
  \HSYK(\mathbf{J}) \;=\; i^{q/2} \!\!\sum_{1\leq i_1<\cdots<i_q\leq N} \!\!J_{i_1\cdots i_q}\,\psi_{i_1}\cdots\psi_{i_q}\,,
\end{equation}
where $q\geq 4$ is an even integer and the couplings $J_{i_1\cdots i_q}$ are independent Gaussian random variables with zero mean and variance $\langle J_{i_1\cdots i_q}^2\rangle = J^2 (q-1)!/N^{q-1}$. The Hilbert space of $N$ Majorana fermions is $\Hil_N\cong\mathbb{C}^{2^{N/2}}$, which decomposes into two sectors of definite fermion parity $(-1)^F$. Throughout this paper we write $L=2^{N/2}$ for the Hilbert space dimension; all results apply equally within a single parity sector of dimension $L/2=2^{N/2-1}$, with the $O(1)$ change in dimension absorbed into the unspecified constants in our bounds. The energy scale is set by $\sigma^2 = J^{2}(q-1)!\binom{N}{q}/N^{q-1}$, and the density of states is well-approximated by a Gaussian $\rho(E)\propto\exp(-E^2/2\sigma^2)$ in the limit of large~$N$. We set $J=1$ throughout. The spectral properties of the SYK model inspire the two key assumptions underlying the findings of this work.

\subsubsection*{Random matrix universality and Haar eigenvector statistics}

The first assumption is that the local eigenvalue statistics of the SYK Hamiltonian, within a fixed fermion parity sector, match those of the appropriate Gaussian random matrix ensemble in the bulk of the spectrum. The specific ensemble depends on the Altland--Zirnbauer symmetry class of the Hamiltonian, which for SYK$_q$ with $q=4$ is determined by $N\bmod 8$ \cite{Altland_1997}. The most consequential manifestation of this universality is the behavior of the disorder-averaged \emph{spectral form factor} (SFF),
\begin{equation}\label{eq:SFF_def}
  \SFF(T) \;=\; \bigl\langle|\tr(e^{-i\HSYK T})|^2\bigr\rangle_{\mathbf{J}}\,,
\end{equation}
which exhibits the universal slope--dip--ramp--plateau structure~\cite{Cotler:2016fpe}. At short times ($T\lesssim\Tdip$), the SFF is dominated by the disconnected part $\SFF_{\mathrm{disc}}(T)\approx L^2 e^{-\sigma^2 T^2}$, which decays as the Fourier transform of the smooth spectral density. This is followed by the dip at $\Tdip\sim\sqrt{N\ln 2}/\sigma$, a linear ramp for $\Tdip\lesssim T\lesssim\THeisen$, and finally the Heisenberg plateau $\SFF(T)\approx L$ for $T\gtrsim\THeisen$, where $\THeisen=2\pi L/(4\sigma)$ is the Heisenberg time. Although the Dyson index~$\beta$ affects the ramp coefficient for a generic random matrix ensemble, Saad, Shenker, and Stanford~\cite{Saad:2018bqo,Saad:2019lba} showed that within a fixed parity sector of SYK, the block and degeneracy structure of the various symmetry classes combines to produce a \emph{universal} ramp $\SFF(T)\approx L\cdot T/\THeisen$, independent of $N\bmod 8$. The analytical agreement between this ramp and the prediction of the JT gravity matrix integral~\cite{Saad:2019lba,Saad:2018bqo} holds to precision $O(e^{-2\Sdim})$ with $\Sdim\sim N$. We refer to the universality of the SFF ramp--plateau structure (independent of Dyson index within a parity sector) as Assumption~A1 (random matrix universality).

For our purpose, which is the Cryptographic censorship in the context of JT/SYK duality, we record the key timescales. Working in units of $1/\sigma$ and considering the Schwarzian regime relevant to the JT gravity application ($\beta J\gg 1$, i.e.\ the low-temperature regime dual to large black holes), we have the hierarchy
\begin{equation}\label{eq:hierarchy_Schw}
  1 \;\ll\; \sigma\,\Tdip \sim \sqrt{N} \;\ll\; \sigma\,\Tscr \sim N^{3/2}\log N \;\ll\; \sigma\,\THeisen \sim 2^{N/2}\,.
\end{equation}

The second assumption concerns the eigenvector statistics of the SYK model. For simple (few-body) operators $O$, the eigenstate thermalization hypothesis (ETH) constrains the matrix elements in the energy eigenbasis to the form~\cite{Sonner:2017hxc,Nayak:2019khe}
\begin{equation}\label{eq:ETH}
  \langle E_a|O|E_b\rangle \;=\; f_O(\bar{E})\,\delta_{ab} + e^{-S(\bar{E})/2}\,g_O(\bar{E},\omega)\,R_{ab}\,,
\end{equation}
where $\bar{E}=(E_a+E_b)/2$, $\omega=E_a-E_b$, $f_O$ and $g_O$ are smooth functions, and $R_{ab}$ are pseudo-random variables with approximately Gaussian statistics. However, in this work, we propose an assumption that extends the standard ETH ansatz, which we refer to as "Haar eigenvector statistics". It states that 
\begin{assumption}[Haar eigenvector statistics]\label{asm:A2}
For independent SYK$_q$ realizations with $q\geq 4$ and $N$ Majorana fermions, the relative eigenbasis rotation $V=V_1^\dagger V_2$ satisfies:
\begin{itemize}[leftmargin=2em, itemsep=3pt]
\item[\textup{}] \emph{Second moment (A2a):} $\bigl|\mathbb{E}_{\mathrm{SYK}}[|V_{ab}|^2]-1/L\bigr|\leq\varepsilon_2/L$ for all $a,b$, with $\varepsilon_2=O(e^{-cN})$.
\item[\textup{}] \emph{Fourth moment (A2b):} $\bigl|\mathbb{E}_{\mathrm{SYK}}[|V_{ab}|^2|V_{a'b'}|^2]-\mathbb{E}_{\mathrm{Haar}}[|V_{ab}|^2|V_{a'b'}|^2]\bigr|\leq\varepsilon_4/L^2$ for all $a,b,a',b'$, with $\varepsilon_4=O(e^{-cN})$.
\end{itemize}
\end{assumption}
Consider two Hamiltonians $\HSYK(\mathbf{J}_1)$ and $\HSYK(\mathbf{J}_2)$ with independent couplings, diagonalized as $\HSYK(\mathbf{J}_i)=V_i D_i V_i^\dagger$. The matrix $V=V_1^\dagger V_2$, whose entries $V_{ab}=\langle E_a^{(1)}|E_b^{(2)}\rangle$ are the overlaps between the two eigenbases, enters the Weingarten calculus through its moments. Assumption~\ref{asm:A2} states that the \emph{relative eigenbasis rotation} between two independent SYK realizations is approximately Haar-distributed at low moment levels. More precisely, condition (A2a) asserts that each entry $|V_{ab}|^2$ has the Haar mean $1/L$ up to exponentially small corrections. Similarly, condition (A2b) tells us that the pairwise correlations between entries also match Haar. Together, they control the second and fourth moments of the entry distribution, which is sufficient for the frame potential analysis of \cref{sec:frame_potential} at all $k=\poly(N)$ through the factorization argument we will describe in Appendix~\ref{app:weingarten}.

Several considerations support Assumption~\ref{asm:A2}. The assumption holds exactly for the appropriate Gaussian ensemble in each symmetry class (GOE, GUE, or GSE depending on $N\bmod 8$), where $V$ is Haar-distributed on the corresponding compact group~\cite{Mehta:2004}; the $O(1/L)$ corrections between $U(L)$-, $O(L)$-, and $Sp(L/2)$-Haar expectations are subleading relative to the $O(e^{-cN})$ tolerance. Additionally, since independent SYK realizations differ by a perturbation of strength~$\sigma$ (comparable to the full spectral width), the eigenvectors are expected to be fully delocalized in each other's basis. This picture has been confirmed rigorously for Wigner matrices with deterministic deformations of comparable spectral scale~\cite{Cipolloni:2024}, where the eigenvector overlaps converge to their Haar values. Lastly, the $k=1$ frame potential formula derived from this assumption matches the direct numerical computation in the slope regime at $N=8$ ($\beta=1$), $N=10$ ($\beta=2$), and $N=12$ ($\beta=4$), as illustrated in \cref{fig:frame_potential}. The late-time $O(1)$ discrepancy decreases with~$N$, which is consistent with $\varepsilon_2=O(e^{-cN})$.
\subsubsection*{The disorder ensemble}

The object of central interest in this paper is the \emph{disorder ensemble} of time evolution operators at fixed time~$T$,
\begin{equation}\label{eq:ensemble_def}
  \Ens_T \;=\; \bigl\{U_{\mathbf{J}} = e^{-i\HSYK(\mathbf{J})T}\bigr\}_{\mathbf{J}}\,,
\end{equation}
where $\mathbf{J}$ is drawn from the standard SYK coupling distribution. This is the ensemble whose proximity to the Haar measure we quantify through the frame potential in \cref{sec:frame_potential} and whose computational indistinguishability we conjecture in \cref{sec:planted_syk}. It should be contrasted with the \emph{time ensemble} generated by a single Hamiltonian $\{e^{-i\HSYK(\mathbf{J}_0)T}\}_{T}$, whose approach to Haar requires the much longer Heisenberg time $\THeisen$~\cite{Cotler:2016fpe,Cotler:2017jue}. The disorder ensemble achieves approximate $k$-design status parametrically faster because the disorder-averaged trace $\langle\tr(e^{-i\HSYK T})\rangle_{\mathbf{J}}=L\,e^{-\sigma^2 T^2/2}$ decays as the characteristic function of the Gaussian spectral density, causing the disconnected SFF to reach $O(1)$ on the timescale $\Tdip\sim\sqrt{N\ln 2}/\sigma$ rather than requiring the Heisenberg time $\THeisen\sim 2^{N/2}/\sigma$ needed for the time ensemble of a single realization.

\subsection{JT gravity and the holographic dictionary}
\label{sec:jt_review}

Jackiw--Teitelboim gravity~\cite{Jackiw:1984je,Teitelboim:1983ux} is a theory of two-dimensional dilaton gravity with action
\begin{equation}\label{eq:JT_action}
  I_{\mathrm{JT}} \;=\; -\Sdim\,\chi(\mathcal{M}) - \frac{1}{2}\int_{\mathcal{M}}\!\sqrt{g}\;\phi\,(R+2) - \int_{\partial\mathcal{M}}\!\sqrt{h}\;\phi\,(K-1)\,,
\end{equation}
where $\phi$ is the dilaton, $R$ is the Ricci scalar, $K$ is the extrinsic curvature at the asymptotic boundary, and $\Sdim$ is the extremal entropy. The equations of motion enforce $R=-2$ (locally AdS$_2$), so all solutions share the same local geometry and differ only in their global causal structure, which is determined by the dilaton profile.

JT gravity arises as the universal low-energy sector of near-extremal black holes~\cite{Maldacena:2016upp,Jensen:2016pah}, and its boundary dynamics is governed by the Schwarzian theory with action $I_{\mathrm{Schw}}=-C\int d\tau\,\{f(\tau),\tau\}$, where $f(\tau)$ is the boundary reparametrization mode, $\{f,\tau\}$ denotes the Schwarzian derivative, and $C=N\alpha_S/J$ is the coupling constant, with $\alpha_S$ a $q$-dependent numerical coefficient determined by UV-IR matching~\cite{Maldacena:2016hyu,Kitaev:2017awl}. At finite~$N$, this Schwarzian sector is the low-energy limit of the SYK model~\cite{Maldacena:2016hyu,Kitaev:2017awl}, establishing the JT/SYK holographic duality.

The implementation of Cryptographic Censorship necessitates that the bulk and boundary theories meet certain structural assumptions, as outlined in~\cite{Engelhardt:2024hpe}. We now highlight and verify each of these in the JT/SYK context.

\subsubsection*{Bulk Hilbert space structure}
We require that the bulk Hilbert space decompose as a direct sum $\Hil=\bigoplus_g\Hil_g$. Here, each sector $\Hil_g$ describes perturbative quantum field theory on a fixed classical background labeled by~$g$, and operators preserve this decomposition, meaning that $\BH=\bigoplus_g\mathcal{B}(\Hil_g)$. In JT gravity, the classical solutions are labeled by the horizon value of the dilaton $\phi_h$, which determines the black hole mass. Each such solution defines a hyperbolic disk geometry with a specific dilaton profile, and the corresponding Hilbert space $\Hil_g$ consists of perturbative matter excitations on that background (or is one-dimensional for pure JT gravity). States in different mass sectors have exponentially suppressed overlap $\sim e^{-\Sdim}$~\cite{Saad:2018bqo,Saad:2019lba}, so the direct-sum structure holds up to corrections that vanish in the large-$N$ limit. Cross-sector amplitudes, computed by the Saad--Shenker--Stanford matrix integral as topology-changing contributions, are suppressed as $O(e^{-2\Sdim})$.

\subsubsection*{The bulk-to-boundary code and the code subspace}\label{sec:jt_gravity}
The next assumption posits the existence of a sequence of bounded linear maps $V_N:\Hil\to\Hil_N$ satisfying $\lim_{N\to\infty}\langle\psi|V_N^\dagger V_N|\varphi\rangle=\langle\psi|\varphi\rangle$ for all $|\psi\rangle,|\varphi\rangle\in\Hil$. This is the statement that the holographic code is \emph{asymptotically isometric}~\cite{Faulkner:2022ada,leutheusser2024subregionsubalgebradualityemergencespace}. In other words, the encoding preserves inner products in the large-$N$ limit, with corrections that vanish as $N\to\infty$. In the JT/SYK duality, the map $V_N$ has been constructed through several complementary approaches. In the energy eigenbasis, $V_N$ sends bulk energy eigenstates $|E\rangle_{\mathrm{bulk}}$ to the corresponding SYK eigenstates $|E_a(N)\rangle_{\mathrm{SYK}}$, with the discrete SYK spectrum approaching the continuous JT gravity spectrum at large~$N$. 

In this work, we take the code subspace to be the microcanonical window
\begin{equation}\label{eq:code_subspace}
  \Hcode \;=\; \mathrm{span}\bigl\{|E_a\rangle : |E_a - E_0|<\Delta E\bigr\}\,,
\end{equation}
with $E_0$ in the bulk of the SYK spectrum (i.e., $|E_0|\lesssim\sigma$, where $\sigma\sim\sqrt{N/q}$ is the energy scale defined in \cref{sec:syk_review}) and width $\Delta E$ chosen so that $\Hcode$ contains exponentially many states while remaining a proper subspace of~$\Hil_N$. The dimension of the code subspace is $\dim\Hcode\sim\rho(E_0)\cdot\Delta E$, where $\rho(E_0)\sim L/\sigma$ is the density of states in the bulk of the spectrum. Typical states in this window are dual to JT black holes with horizon dilaton value $\phi_h\propto\sqrt{E_0-E_{\mathrm{gs}}}$, where $E_{\mathrm{gs}}$ is the ground state energy~\cite{Maldacena:2016upp}.

Moreover, we need that that bulk and boundary time evolution agree on the code subspace. Explicitly, we have $V_N^\dagger\, e^{-iH_N t}\,V_N = e^{-iH_{\mathrm{bulk}}t}$, where $H_N=\HSYK$ and $H_{\mathrm{bulk}}$ is the bulk ADM Hamiltonian. Since the boundary graviton mode of JT gravity is the Schwarzian, whose Hamiltonian $\HSchw$ equals the ADM Hamiltonian~\cite{Maldacena:2016upp}, and the Schwarzian theory is the low-energy limit of SYK, this identity holds in the energy eigenbasis with corrections of order $1/N$ to the individual eigenvalues. For boundary times $t=O(N^0)$, the accumulated error $\|V_N^\dagger\,e^{-iH_N t}\,V_N-e^{-iH_{\mathrm{bulk}}t}\|$ is $O(1/N)$ from perturbative corrections beyond the Schwarzian, with non-perturbative corrections of order $O(e^{-\Sdim})$~\cite{Saad:2019lba}.

\subsection{Pseudorandom unitaries and Cryptographic Censorship}
\label{sec:pru_review}

We now recall the complexity-theoretic definitions and the main theorem from~\cite{Engelhardt:2024hpe} that we will apply to the JT/SYK system. The central notion is that of a pseudorandom unitary ensemble---an efficiently generated family of unitaries that no polynomial-time quantum algorithm can distinguish from Haar-random.

\begin{definition}[PRU~{\cite{cryptoeprint:2018/544,Engelhardt:2024hpe}}]\label{def:pru}
An ensemble of unitaries $\mathcal{U}=\{U_k\in U(\Hil)\}_{k\in\mathcal{K}}$ indexed by a key space $\mathcal{K}$ is a \emph{pseudorandom unitary ensemble} (PRU) with security parameter~$\kappa$ if the following conditions hold. 

There exists a $\poly(\kappa)$-time quantum algorithm $Q$ such that
\[
Q(k,|\psi\rangle)=U_k|\psi\rangle
\]
for every $k\in\mathcal{K}$ and $|\psi\rangle\in\Hil$ (efficient generation). 

Moreover, for every quantum algorithm~$\mathcal{A}$ running in time $\poly(\kappa)$ and making at most $m=\poly(\kappa)$ adaptive oracle queries to $U$ and $U^\dagger$,
\begin{equation}\label{eq:pru_def}
  \bigl|\Pr_{k\leftarrow\mathcal{K}}\!\bigl[\mathcal{A}^{U_k}=1\bigr]
  - \Pr_{U\leftarrow\mu}\!\bigl[\mathcal{A}^{U}=1\bigr]\bigr|
  \;=\; \negl(\kappa)\,,
\end{equation}
where $\mu$ denotes the Haar measure on $U(\Hil)$ (computational indistinguishability).
\end{definition}

Computational indistinguishability is strictly stronger than statistical moment-matching. An $\varepsilon$-approximate unitary $k$-design satisfies \cref{eq:pru_def} for the restricted class of algorithms that query $U$ non-adaptively on a fixed $m$-partite state, but a PRU must defeat all efficient adversaries, including those employing adaptive strategies, entanglement, and arbitrary intermediate computation. The gap between these notions is precisely the gap between information-theoretic and computational security in cryptography, and bridging it for the SYK ensemble is the purpose of the planted-SYK conjecture in \cref{sec:planted_syk}.

\begin{definition}[GPRU~{\cite{Engelhardt:2024hpe}}]\label{def:gpru}
A PRU is \emph{gravitationally pseudorandom} (GPRU) if its security parameter scales polynomially with the gravitational coupling: $\kappa=\poly(\GN^{-1})$, where $\GN$ is the bulk Newton constant.
\end{definition}

In the JT/SYK setting, $\GN\sim 1/N$ and $\kappa=N/2$, so the GPRU requirement reduces to $\kappa=\poly(N)$. A GPRU also generates pseudorandom state ensembles. Concretely, if $\mathcal{U}=\{U_k\}$ is a GPRU, then $\Psi=\{U_k|\psi_0\rangle\}$ is a pseudorandom state ensemble (PRS) for any fixed $|\psi_0\rangle\in\Hcode$, meaning that $\poly(\kappa)$ copies of the state cannot be distinguished from Haar-random states by any efficient algorithm.

The Cryptographic Censorship theorem connects gravitational pseudorandomness to the geometry of the bulk spacetime. Its proof integrates a quantum learning no-go theorem~\cite{Yang:2023ddi} with a measure concentration argument on the unitary group to establish that a large region of the bulk must be causally inaccessible from the boundary. We state the theorem in the form most directly applicable to our setting, paraphrasing from~\cite{Engelhardt:2024hpe}.

\begin{theorem}[Cryptographic Censorship~{\cite{Engelhardt:2024hpe}}]\label{thm:cc_general}
Let $\Hil=\bigoplus_g\Hil_g$ be the bulk Hilbert space with asymptotically isometric code maps $V_N:\Hil\to\Hil_N$ satisfying $V_N^\dagger\,e^{-iH_N t}\,V_N=e^{-iH_{\mathrm{bulk}}t}$. Let $U_N$ be a unitary on $\Hil_N$ that is typical in a GPRU ensemble, and let $|\psi_N(t)\rangle\in\Hil_N$ be a state whose complete time evolution is dual to a strongly causal geometric bulk $(M,g)$. Suppose there exists a $\poly(\kappa)$-time, $({}$randomized${}$\,$)$ $K$-Lipschitz algorithm $\mathcal{A}$ that reconstructs the causal wedge of the asymptotic boundary $\mathscr{I}$, and that $U_N|\psi_N(t)\rangle$ is typical in the corresponding PRS.

Then $(M,g)$ contains an event horizon: $\Ncw\subsetneq\BH$.
\end{theorem}

The theorem works by contradiction, through a chain of reasoning we now sketch. The learning no-go theorem of Yang and Engelhardt~\cite{Yang:2023ddi} guarantees that no efficient algorithm can predict the full action of a pseudorandom unitary with high fidelity. Specifically, for any $\poly(\kappa)$-time algorithm $\mathcal{A}$ producing an approximation~$\Vhat$, we have 
\begin{equation}\label{eq:learning_nogo}
  \underset{\Vhat\leftarrow\mathcal{A}^{U}}{\avg}\; F\!\bigl(\Vhat|\psi\rangle,\,U|\psi\rangle\bigr) \;\leq\; 1-\alpha\,,
\end{equation}
where $F$ denotes the fidelity, $\alpha>0$ is independent of $\kappa$, and the bound holds for typical~$U$ in the GPRU and typical~$|\psi\rangle$ in the PRS. Since the fidelity falls short of unity by an $O(1)$ amount, the trace distance between the predicted and true output states is bounded below. By the operational interpretation of trace distance, there must exist an observable $Q_*$ with $\|Q_*\|_\infty\leq 1$ for which the algorithm's prediction deviates from the true expectation value. We obtain
\begin{equation}\label{eq:distinguishing_gap}
  \underset{\Vhat}{\avg}\;\bigl|\tr\bigl(Q_*\,\Vhat\,\psi\,\Vhat^\dagger\bigr)-\tr\bigl(Q_*\,U\,\psi\,U^\dagger\bigr)\bigr| \;\geq\; \sqrt{\alpha} - O\!\bigl(1/\poly(\kappa)\bigr)\,,
\end{equation}
where $\alpha>0$ is the fidelity deficit from \cref{eq:learning_nogo} and the $O(1/\poly(\kappa))$ correction absorbs the errors from approximate norm preservation and the typicality concentration bound.
If the algorithm $\mathcal{A}$ correctly reconstructs all operators in the causal wedge algebra~$\Ncw$, then $Q_*$ cannot belong to~$\Ncw$, and the causal wedge is a proper subset of the full bulk operator algebra. The boundary of the causal wedge of~$\mathscr{I}$ is, by definition, the event horizon, so its nontriviality implies the existence of a horizon. In \cref{sec:distinguishing_operator} we will identify $Q_*$ with the regularized geodesic length of the maximal interior slice ~$\ell$ and compute the gap in \cref{eq:distinguishing_gap}.

\section{SYK as a pseudorandom unitary ensemble}
\label{sec:syk_pru}

This section contains the technical core of the paper. We first prove that the SYK disorder ensemble $\Ens_T$ defined in \cref{eq:ensemble_def} is an approximate unitary $k$-design for all $k=\poly(N)$ (\cref{sec:frame_potential}), then formulate the planted-SYK conjecture and present evidence for its plausibility (\cref{sec:planted_syk}), and finally combine these into the conditional theorem that $\Ens_T$ is a GPRU (\cref{sec:gpru_theorem}).

\subsection{Frame potential and the Weingarten calculus}
\label{sec:frame_potential}

The standard diagnostic for the distance between a unitary ensemble and the Haar measure is the \emph{frame potential}. For an ensemble $\Ens$ on $U(L)$, the $k$-th frame potential is defined as
\begin{equation}\label{eq:frame_potential_def}
  \Fk(\Ens) \;=\; \underset{U_1,U_2\sim\Ens}{\mathbb{E}}\;\bigl|\tr(U_1^\dagger U_2)\bigr|^{2k}\,.
\end{equation}
A theorem of Scott~\cite{Scott:2008} (see also \cite{PhysRevA.80.012304}) guarantees that $\Fk(\Ens)\geq k!$ for any ensemble, with equality if and only if $\Ens$ is an exact unitary $k$-design. An ensemble satisfying $\Fk(\Ens)\leq k!(1+\varepsilon)$ is an $\varepsilon$-approximate $k$-design. Our goal is to evaluate \cref{eq:frame_potential_def} for $\Ens=\Ens_T$ and show that the deviation from $k!$ is negligible for all $k=\poly(N)$.

By employing the diagonalization $\HSYK(\mathbf{J}_i)=V_i D_i V_i^\dagger$ described earlier with $D_i=\mathrm{diag}(E_1^{(i)},\ldots,E_L^{(i)})$, the time evolution becomes $U_i=V_i\,\mathrm{diag}(e^{-iE_a^{(i)}T})\,V_i^\dagger$, resulting in
\begin{equation}\label{eq:trace_decomp}
  \tr(U_1^\dagger U_2) \;=\; \sum_{a,b=1}^L e^{i E_a^{(1)}T}\,e^{-i E_b^{(2)}T}\,|V_{ab}|^2\,.
\end{equation}
We compute $\Fk(T)$ by first averaging over the eigenvalues (spectral average) and then over $V$ (eigenvector average), replacing the latter with a Haar integral.

For the spectral average, we define the \emph{spectral power sums}
\begin{equation}\label{eq:spectral_power}
  p_n^{(i)}(T) \;=\; \sum_{a=1}^L e^{-in\,E_a^{(i)}\,T} \;=\; \tr\!\bigl(e^{-in\HSYK(\mathbf{J}_i)T}\bigr)\,,
\end{equation}
so that $\SFF(T)=\langle|p_1|^2\rangle_{\mathbf{J}}$, which connects directly to \cref{eq:SFF_def}. Expanding $|\tr(U_1^\dagger U_2)|^{2k}$ and performing the Haar average over $V$ requires the unitary Weingarten calculus~\cite{Collins:2006,Collins:2022}, which expresses moments of Haar-distributed matrices in terms of permutations. The general formula takes the form
\begin{equation}\label{eq:Fk_weingarten}
  \Fk(T) \;=\; \sum_{\sigma,\tau\in S_{2k}} \Wg(\sigma^{-1}\tau,\,L)\;\bigl\langle C_\sigma^{(1)}\bigr\rangle_{\!\mathbf{J}_1}\;\bigl\langle C_\tau^{(2)}\bigr\rangle_{\!\mathbf{J}_2}^{\!*}\,,
\end{equation}
where the sum runs over the symmetric group $S_{2k}$, $\Wg(\pi,L)$ is the Weingarten function for $U(L)$ at tensor power $2k$, and $C_\sigma^{(i)}$ is a spectral factor associated to the permutation $\sigma$. Each cycle $c$ of $\sigma$ contributes a factor $p_{n_c}^{(i)}(T)$ to $C_\sigma^{(i)}$, where the cycle power $n_c=\sum_{m\in c}s_m$ is determined by the sign pattern $\mathbf{s}=(\underbrace{+1,\ldots,+1}_{k},\underbrace{-1,\ldots,-1}_{k})$ that distinguishes the $k$ holomorphic from the $k$ antiholomorphic copies in $|z|^{2k}$. We collect the derivation of \cref{eq:Fk_weingarten} in Appendix~\ref{app:weingarten}.

\subsubsection*{The $k=1$ formula}
For $k=1$, the sign pattern is $\mathbf{s}=(+1,-1)$ and the symmetric group is $S_2=\{\mathrm{id},(12)\}$. The identity permutation has two singleton cycles with powers $n_1=+1$ and $n_2=-1$, giving
\begin{equation}\label{eq:k1_spectral_id}
  C_{\mathrm{id}} \;=\; p_{+1}\,p_{-1} \;=\; |p_1|^2 .
\end{equation}
The transposition $(12)$ has a single 2-cycle containing both a $+1$ and a $-1$ index, so $n_c=+1-1=0$ and
\begin{equation}\label{eq:k1_spectral_trans}
  C_{(12)} \;=\; p_0 \;=\; L\,.
\end{equation}
The Weingarten values are $\Wg(\mathrm{id},L)=1/(L^2-1)$ and $\Wg((12),L)=-1/(L(L^2-1))$. Since the SFF and the Weingarten function factorize across the two independent realizations $\mathbf{J}_1,\mathbf{J}_2$, the sum \cref{eq:Fk_weingarten} over all four $(\sigma,\tau)$ pairs evaluates to
\begin{equation}\label{eq:k1_derivation}
  \Fone(T) \;=\; \frac{\SFF^2 - 2\,\SFF\cdot L\cdot\frac{1}{L} + L^2}{L^2-1} \;=\; \frac{(\SFF-1)^2}{L^2-1}+1\,,
\end{equation}
reproducing \cref{eq:F1_intro}. The cross terms arise from the $(\sigma,\tau)=(\mathrm{id},(12))$ and $((12),\mathrm{id})$ pairs, each contributing $-\SFF\cdot L/(L(L^2-1))=-\SFF/(L^2-1)$, while the $(\sigma,\tau)=((12),(12))$ pair contributes $L^2\cdot\Wg(\mathrm{id},L)=L^2/(L^2-1)$. The formula makes transparent the relationship between approximate $1$-design status and the spectral form factor: $\Fone\to 1$ precisely when $\SFF(T)\to 1+O(\sqrt{L})$, which occurs at the beginning of the ramp regime $T\gtrsim\Tdip$. We have verified \cref{eq:F1_intro} numerically by computing both sides independently for $N=8,10,12$ with $100$--$300$ disorder realizations per system size. The result is shown in \cref{fig:frame_potential}.
\begin{figure}[t]
  \centering
  \includegraphics[width=\textwidth]{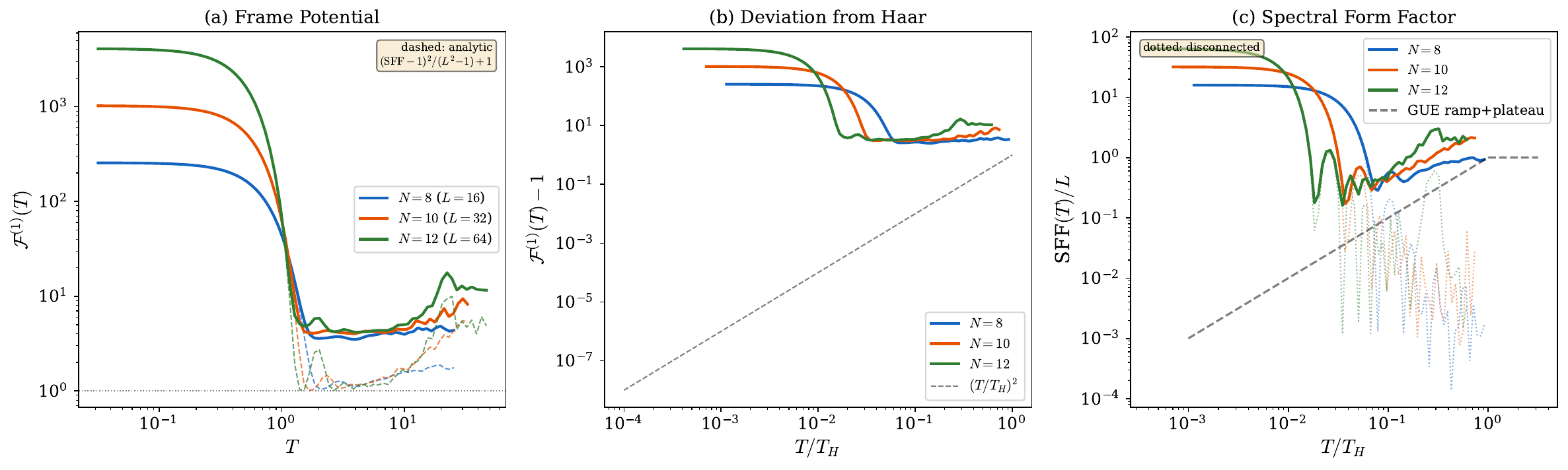}
  \caption{Numerical verification of the frame potential analysis for SYK$_4$ with $N=8,10,12$ ($L=16,32,64$), averaged over $100$--$300$ disorder realizations. \textbf{(a)}~The first frame potential $\Fone(T)$ computed directly from cross-realization averages of $|\tr(U_1^\dagger U_2)|^2$ (solid), compared with the analytic formula \cref{eq:F1_intro} evaluated from the disorder-averaged SFF (dashed). \textbf{(b)}~The deviation $\Fone(T)-1$ as a function of $T/\THeisen$, showing the approach to the $\tau^2=(T/\THeisen)^2$ scaling predicted by \cref{eq:kdesign_bound} (dashed gray line). \textbf{(c)}~The normalized spectral form factor $\SFF(T)/L$, exhibiting the universal slope--dip--ramp--plateau structure. Solid lines show the connected SFF; dotted lines show the disconnected part.}
  \label{fig:frame_potential}
\end{figure}

\subsubsection*{Paired and unpaired permutations}

To generalize to arbitrary~$k$, we classify the permutations $\sigma\in S_{2k}$ according to their spectral factors. A permutation is \emph{paired} if every cycle has zero total power ($n_c=0$), meaning each cycle contains equal numbers of $+1$ and $-1$ indices from the sign pattern~$\mathbf{s}$. For a paired permutation, $C_\sigma=L^{c(\sigma)}$, independent of the eigenvalue data, where $c(\sigma)$ is the number of cycles. An \emph{unpaired} permutation has at least one cycle with $n_c\neq 0$, so that $C_\sigma$ depends on the spectral power sums $p_{n_c}$ and decays in the ramp regime.

The key structural result is that the contribution of all paired permutations to the Weingarten sum in \cref{eq:Fk_weingarten} gives exactly the Haar frame potential:
\begin{equation}\label{eq:paired_sum}
  \sum_{\substack{\sigma,\tau\in S_{2k}\\\text{both paired}}} \Wg(\sigma^{-1}\tau,L)\,L^{c(\sigma)+c(\tau)} \;=\; k!\,.
\end{equation}
To see this, consider the Haar frame potential $\Fk_{\mathrm{Haar}}=\mathbb{E}_{W\sim\mathrm{Haar}}[|\tr W|^{2k}]=k!$ for $k\leq L$, and decompose it using the eigenbasis factorization of this section. For the Haar (CUE) eigenphase distribution, $\langle p_n\rangle_{\mathrm{CUE}}=0$ for all $n\neq 0$ by the rotational symmetry of the eigenphase distribution on the unit circle. Since the spectral factor $C_\sigma$ of any unpaired permutation contains at least one power sum $p_{n_c}$ with $n_c\neq 0$, and the two realizations are independent, the disorder average $\langle C_\sigma^{(1)}\rangle_{\mathrm{CUE}}\langle (C_\tau^{(2)})^*\rangle_{\mathrm{CUE}}$ vanishes whenever $\sigma$ or $\tau$ is unpaired. Only the paired terms survive the eigenvalue average, and their sum must equal $k!$.

\subsubsection*{The general bound}

The dominant unpaired contribution comes from permutations with the \emph{maximum} number of cycles, $c(\sigma)=k+1$: two singleton cycles of powers $n_c=+1$ and $n_c=-1$ (each consisting of a single index), together with $k-1$ paired two-cycles. There are $k\cdot k!$ such permutations, counted as $k$ choices for the positive singleton among the $k$ positive indices, $k$ choices for the negative singleton among the $k$ negative indices, and $(k-1)!$ perfect matchings of the remaining $2(k-1)$ indices into paired two-cycles. Each contributes a spectral factor $L^{(k+1)-2}\cdot\SFF = L^{k-1}\cdot\SFF$, and the leading Weingarten term ($\sigma=\tau$) gives $\Wg(\mathrm{id},L)\sim L^{-2k}$, so the net contribution is $L^{2(k-1)-2k}\cdot\SFF^2=\SFF^2/L^2=\tau^2$ per permutation. Minimally unpaired permutations with longer unpaired cycles (e.g., a three-cycle of power $+1$) also exist, but they necessarily have fewer total cycles and their contributions are suppressed by at least $1/L^2$; see Appendix~\ref{app:weingarten} for details. Summing the $k\cdot k!$ leading terms and the subleading corrections:

\begin{theorem}[Approximate $k$-design]\label{thm:kdesign}
Let $\Ens_T$ be the SYK$_q$ disorder ensemble with $q\geq 4$ and $N$ Majorana fermions, Hilbert space dimension $L=2^{N/2}$, energy scale $\sigma^2=J^2(q-1)!\binom{N}{q}/N^{q-1}$, and Heisenberg time $\THeisen=2\pi L/(4\sigma)$. Define $\tau=T/\THeisen$ and $\Tdip=\sqrt{N\ln 2}/\sigma$. Under Assumption A1~(random matrix universality) and Assumption~\ref{asm:A2}~(Haar eigenvector statistics), for any $k\leq N^a$ with $a>0$ fixed and any $T$ satisfying $\Tdip\lesssim T\leq\THeisen$,
\begin{equation}\label{eq:kdesign_bound}
  \frac{|\Fk(T)-k!|}{k!} \;\leq\; C\cdot k\cdot\tau^2 + O\!\bigl(\poly(k)\cdot e^{-cN}\bigr)\,,
\end{equation}
where $C,c>0$ are constants independent of $k$, $N$, and $T$. The first term is the spectral ramp contribution. The second term is the eigenvector correction: under (A2a) alone it takes the form $O(k\cdot e^{-cN})$, and under both (A2a) and (A2b) the polynomial prefactor in~$k$ is at most $O(k^2)$. In particular, for $k=\poly(N)$ and any $T=\poly(N)$, the right-hand side is $\negl(N)$, and $\Ens_T$ is an $\varepsilon$-approximate unitary $k$-design with $\varepsilon=\negl(N)$.
\end{theorem}

\noindent The eigenvector correction in \cref{eq:kdesign_bound} requires only the low-moment conditions (A2a) and (A2b) of Assumption~\ref{asm:A2}, not Haar statistics at the $(2k)$-th moment level. This is because the Weingarten decomposition separates the $k$-th frame potential into \emph{paired} contractions, which use only the unitarity constraint $VV^\dagger=I$ and carry no eigenvector correction, and \emph{unpaired} contractions, where each of the two unpaired singletons contributes a correction of $O(\varepsilon_2)$ from (A2a). The cross-correlations between paired and unpaired contractions are controlled by the 4th-moment condition (A2b), contributing $O(\varepsilon_4)$ per cross-term. Since there are $k\cdot k!$ leading terms with $O(k)$ cross-correlations each, the total correction normalized by $k!$ is $O(k^2\cdot e^{-cN})=\negl(N)$ for $k=\poly(N)$. A detailed derivation is given in the Appendix~\ref{app:weingarten}.

\subsection{The planted-SYK conjecture}
\label{sec:planted_syk}

Theorem~\ref{thm:kdesign} rules out all distinguishers that rely on the first $\poly(N)$ moments of the unitary distribution, but a PRU must resist all efficient adversaries, including those that exploit structural information beyond moments. To close this gap, we propose a computational hardness assumption specific to the SYK model.

\subsubsection*{The distinguishing problem}

The relevant computational problem is an instance of a planted random matrix problem: given oracle access to a unitary, determine whether it was drawn from a structured ensemble or from the Haar measure.

\begin{definition}[Planted-SYK distinguishing]\label{def:psyk}
Let $N$, $q\geq 4$, $T>0$, and $m=\poly(N)$ be given. An instance of the planted-SYK problem consists of oracle access to a unitary $U$ acting on $N/2$ qubits, drawn from one of two distributions:
\begin{equation}\label{eq:psyk_hypotheses}
  H_0:\; U\leftarrow\mu\bigl(U(L)\bigr) \qquad\text{vs.}\qquad H_1:\; U=e^{-i\HSYK(\mathbf{J})T},\;\;\mathbf{J}\sim\text{standard SYK}_q\,.
\end{equation}
A quantum algorithm $\mathcal{A}$ making at most $m$ adaptive queries to $U$ and $U^\dagger$, with $\poly(N)$-time intermediate computation, has \emph{advantage} $\varepsilon=|\Pr[\mathcal{A}^U\!=\!1\mid H_0]-\Pr[\mathcal{A}^U\!=\!1\mid H_1]|$.
\end{definition}

We now introduce a conjecture which asserts that no efficient algorithm achieves non-negligible advantage in the physically relevant time window.

\begin{conjecture}[Planted-SYK hardness]\label{conj:planted_syk}
For any $q\geq 4$ (even) and any $T$ satisfying $\Tscr(N)\leq T\leq\exp(N^\gamma)$ for some fixed $0<\gamma<1$, every quantum algorithm of $\poly(N)$ complexity with at most $m=\poly(N)$ adaptive oracle queries achieves advantage at most $\negl(N)$ in the planted-SYK problem of Definition~\ref{def:psyk}.
\end{conjecture}

\noindent The time window in Conjecture~\ref{conj:planted_syk} concerns the \emph{oracle model}. More precisely, the adversary is given black-box access to $U=e^{-i\HSYK T}$ and does not need to generate~$U$. Furthermore, the upper bound $T\leq\exp(N^\gamma)$ prevents the adversary from resolving the full eigenvalue spectrum via phase estimation (which would require $T\gtrsim\THeisen\sim\exp(\Omega(N))$), while the lower bound $T\geq\Tscr$ ensures that scrambling has erased local signatures of the Hamiltonian structure.

We now present evidence for the conjecture.

\subsubsection*{Spectral evidence}

With $m=\poly(N)$ queries to $U=e^{-i\HSYK T}$ at time $T=\poly(N)$, phase estimation determines the eigenvalues $E_a$ to precision $\delta E\sim 1/(mT)=1/\poly(N)$, far coarser than the mean level spacing $\delta\sim\sigma/L\sim\sqrt{N}/2^{N/2}$. Consequently, the adversary cannot resolve individual level spacings or detect the Dyson index ($N\bmod 8$ dependence~\cite{You:2016ldz,Cotler:2016fpe,Behrends:2019}), and the SFF ramp coefficient is in any case independent of the symmetry class within a fixed parity sector~\cite{Saad:2018bqo}. Moreover, the adversary cannot exploit the non-semicircular (Gaussian) shape of the SYK spectral density, because it observes eigenphases $\theta_a=E_a T\!\!\pmod{2\pi}$ that wrap around the unit circle $O(\sigma T)$ times, rendering the phase density approximately uniform. Furthermore, the two-point eigenvalue correlator matches random matrix predictions to $O(e^{-2\Sdim})$ through the JT gravity matrix integral~\cite{Saad:2018bqo,Saad:2019lba}. Likewise, the eigenvector statistics are expected to match Haar with corrections exponentially small in~$N$~\cite{Sonner:2017hxc,Nayak:2019khe,Gharibyan:2018jrp}, as formalized in Assumption~\ref{asm:A2}.

\subsubsection*{The low-degree polynomial framework}

Our next rationale for motivating this conjecture comes from the low-degree polynomial (LDP) framework~\cite{Hopkins:2017cle,Kunisky:2019bwz}. LDP predicts computational phase transitions in hypothesis testing by bounding the power of degree-$D$ polynomial tests on the observed data. We apply it to the classical eigenphase data extracted by phase estimation. The degree-$d$ component of the likelihood ratio between $H_1$ (SYK) and $H_0$ (CUE) probes the $d$-point connected eigenphase correlator. By the genus expansion of the JT gravity matrix integral~\cite{Saad:2018bqo,Saad:2019lba}, the leading correction to the disk-level correlator at genus one contributes at order $O(e^{-d\,c\,N})$ for the $d$-point function, and the cumulant expansion introduces a combinatorial prefactor of $O(d!)$, giving the formal estimate $|\Delta R_d^{\mathrm{conn}}|\lesssim d!\cdot e^{-d\,c\,N}$. The number of degree-$d$ monomials in $M=\poly(N)$ eigenphases is $O(\poly(N)^d)$, so
\begin{equation}\label{eq:ldp_bound}
  \mathrm{SNR}_D \;\leq\; \exp\!\bigl(O(D\log N)-D\,c\,N\bigr)\,.
\end{equation}
For $D=O(\log N)$, this gives $\mathrm{SNR}_D\leq\exp(-\Omega(N\log N))$. By the LDP meta-conjecture---validated for planted clique, sparse PCA, and community detection---failure of degree-$O(\log n)$ tests predicts failure of all polynomial-time algorithms. However, it should be emphasized that the LDP framework is classical and does not directly bound quantum oracle algorithms; extending it to the quantum setting remains an important open problem. In light of this crucial fact, we treat this analysis as complementary (or speculative) evidence for Conjecture~\ref{conj:planted_syk}, alongside the spectral universality and $k$-design results.

\subsubsection*{Comparison with established planted problems}

To calibrate the strength of the evidence, we compare the planted-SYK problem with two benchmark problems whose computational hardness is widely believed. In the \emph{planted clique} problem, a clique of size~$k$ hidden in $G(n,1/2)$ produces a rank-one perturbation of the adjacency matrix with an outlier eigenvalue detectable by PCA when $k\geq\sqrt{n}$; below this threshold, no polynomial-time algorithm is known and the LDP framework confirms $\mathrm{SNR}=o(1)$~\cite{Hopkins:2017cle}. Additionally, in the \emph{spiked random matrix} (sparse PCA) problem, a rank-one perturbation $\lambda vv^\top$ added to a Wigner matrix produces an outlier eigenvalue above the Baik--Ben~Arous--P\'{e}ch\'{e} threshold at $\lambda=1$~\cite{Baik:2005}; for sparse~$v$ and $\lambda<1$, the problem is conjectured to be computationally hard despite being information-theoretically solvable. In both cases, the planted-SYK problem lacks the structural features that enable tractability: the spectral deviation from CUE is full-rank, has amplitude $O(e^{-cN})$, and is spread uniformly across the spectrum with no outlier or localized signal. Therefore, the LDP bound \cref{eq:ldp_bound} indicates that the signal is exponentially weaker than in either benchmark at their respective computational thresholds.

\subsection{SYK as a candidate GPRU under planted-SYK hardness}\label{sec:gpru_theorem}
We now combine \cref{thm:kdesign} with Conjecture~\ref{conj:planted_syk} to establish the main result of this section: the SYK disorder ensemble is a GPRU.

\begin{theorem}[SYK as a GPRU]\label{thm:gpru}
Assume A1~(random matrix universality), Haar eigenvector statistics~(A2), and Conjecture~\ref{conj:planted_syk} (planted-SYK hardness). For any $q\geq 4$ (even) and any $T=\poly(N)$ satisfying $T\geq\Tscr(N)$, the SYK disorder ensemble $\Ens_T=\{e^{-i\HSYK(\mathbf{J})T}\}_{\mathbf{J}\in\mathcal{K}}$ is a gravitationally pseudorandom unitary ensemble in the sense of Definition~\ref{def:gpru}, with security parameter $\kappa=N/2$.
\end{theorem}

\begin{proof}
We proceed by verifying the requirements of Definition~\ref{def:pru}. For efficient generation, the key $k\in\mathcal{K}$ encodes the $\binom{N}{q}$ SYK couplings $\mathbf{J}$ discretized to $B=\poly(N)$ bits of precision. Given~$k$ and a quantum state~$|\psi\rangle$, the algorithm $Q(k,|\psi\rangle)=e^{-i\HSYK(\mathbf{J})T}|\psi\rangle$ is implemented by Trotter--Suzuki decomposition with circuit depth $O(T^2 N^{2q}/\varepsilon)$, which is $\poly(N)$ for $T=\poly(N)$ and $\varepsilon=1/\poly(N)$ (see Appendix~\ref{app:trotter}). The discretization error $\|U_{\mathbf{J}}-U_{\tilde{\mathbf{J}}}\|_\infty\leq T\cdot N^q\cdot 2^{-B}=\negl(N)$ is also negligible.  

For computational indistinguishability, since $T=\poly(N)$ and $T\geq\Tscr$, the time~$T$ lies in the window $[\Tscr,\exp(N^\gamma)]$ of Conjecture~\ref{conj:planted_syk} for any $\gamma>0$. The conjecture then directly implies that every $\poly(N)$-time adversary achieves advantage $\negl(N)$, which is condition \cref{eq:pru_def} of Definition~\ref{def:pru}. Also, since $\kappa=N/2$ and $\GN\sim 1/N$, we have $\kappa=\poly(\GN^{-1})$, which satisfies Definition~\ref{def:gpru}. Accordingly, the ensemble $\Psi_T=\{e^{-i\HSYK(\mathbf{J})T}|\psi_0\rangle\}_{\mathbf{J}\in\mathcal{K}}$ is a pseudorandom state ensemble for any fixed initial state $|\psi_0\rangle\in\Hcode$.
\end{proof}

We now record immediate corollaries that will also be utilized in \cref{sec:crypto_censorship}.

\begin{corollary}[Learning hardness]\label{cor:learning}
Under the assumptions of Theorem~\ref{thm:gpru}, the SYK time evolution satisfies the learning no-go theorem of Yang and Engelhardt~\cite{Yang:2023ddi}: for any $\poly(N)$-time quantum algorithm $\mathcal{A}$ with oracle access to $U_{\mathbf{J}}=e^{-i\HSYK(\mathbf{J})T}$,
\begin{equation}\label{eq:fidelity_bound}
  \underset{\mathbf{J},\,\psi,\,\Vhat\leftarrow\mathcal{A}^{U_{\mathbf{J}}}}{\avg}\; F\!\bigl(\Vhat|\psi\rangle,\,U_{\mathbf{J}}|\psi\rangle\bigr) \;\leq\; 1-\alpha\,,
\end{equation}
where $\alpha>0$ is a constant independent of~$N$.
\end{corollary}

This is the specialization of \cref{eq:learning_nogo} to the SYK ensemble, and it ensures that the fidelity gap needed for Cryptographic Censorship is present.

\begin{corollary}[Distinguishing operator]\label{cor:distinguishing}
Under the assumptions of Theorem~\ref{thm:gpru}, for any $\poly(N)$-time algorithm $\mathcal{A}$ that correctly predicts all operators in the causal wedge algebra $\Ncw$, there exists an observable $Q_*$ with $\|Q_*\|_\infty\leq 1$ and $Q_*\notin\Ncw$ such that
\begin{equation}\label{eq:Q_star_gap}
  \underset{\Vhat\leftarrow\mathcal{A}^{U}}{\avg}\;\bigl|\tr\!\bigl(Q_*\,\Vhat\,\psi\,\Vhat^\dagger\bigr)-\tr\!\bigl(Q_*\,U\,\psi\,U^\dagger\bigr)\bigr| \;\geq\; \sqrt{\alpha} - O\!\bigl(1/\poly(N)\bigr)\,.
\end{equation}
\end{corollary}

This is the operator-level consequence of Corollary~\ref{cor:learning}. It states that since the algorithm achieves perfect prediction on~$\Ncw$ but imperfect fidelity overall, the deficit must be concentrated on operators outside the causal wedge. The identification of~$Q_*$ with the regularized geodesic length~$\ell$ of the maximal interior slice is the subject of \cref{sec:distinguishing_operator}.

\noindent\textbf{Single realization versus disorder average.} An important conceptual distinction separates the $k$-design result (Theorem~\ref{thm:kdesign}) from the GPRU conclusion (Theorem~\ref{thm:gpru}). The $k$-design property is a statement about the disorder \emph{ensemble}: it controls the moments $\mathbb{E}_{\mathbf{J}}[f(U_{\mathbf{J}})]$ and handles non-adaptive adversaries who query $U^{\otimes m}$ on a fixed entangled state. The PRU adversary, by contrast, receives a \emph{single} $U=e^{-i\HSYK(\mathbf{J}_0)T}$ for one fixed coupling vector~$\mathbf{J}_0$ and may query it adaptively. A single SYK unitary has a fixed spectrum and conserved energy---structure absent from a Haar-random unitary. However, this structure is inaccessible to the adversary for three reasons. First, the energy eigenstates of $\HSYK(\mathbf{J}_0)$ are fully scrambled after $T\geq\Tscr$: by ETH, they are indistinguishable from random states in any few-body basis, so the conserved-energy structure is hidden behind the delocalization of the eigenbasis. Second, the eigenphases $\theta_a=E_a T\!\!\pmod{2\pi}$ that the adversary can extract via phase estimation are wrapped $O(\sigma T)\gg 1$ times around the unit circle, and a $\poly(N)$-sized subsample of wrapped eigenphases from the SYK spectrum is indistinguishable from a subsample of CUE eigenphases---this is the content of Conjecture~\ref{conj:planted_syk}. Third, the Meckes concentration inequality \cref{eq:concentration} ensures that the fidelity function $F(U)$ concentrates around its ensemble mean for a $1-\exp(-c\cdot 2^{N/2}/\poly(N))$ fraction of coupling vectors~$\mathbf{J}_0$, so the ensemble average faithfully represents the typical single realization. It is this combination of ETH (hiding the eigenbasis), eigenphase wrapping (hiding the spectral density), and measure concentration (bridging ensemble to single realization) that makes Conjecture~\ref{conj:planted_syk} a natural companion to the $k$-design result, upgrading statistical moment-matching to full computational indistinguishability.

\subsection{Fidelity deficit from the \texorpdfstring{$k$}{k}-design property}\label{sec:fid_defi}

We now show that the $k$-design result of Theorem~\ref{thm:kdesign} implies a near-maximal fidelity deficit for non-adaptive adversaries, without invoking Conjecture~\ref{conj:planted_syk}.

\begin{proposition}\label{prop:fidelity_nonadaptive}
Let $\Ens_T$ be an $\varepsilon$-approximate unitary $k$-design on $\Hil\cong\mathbb{C}^L$ with $k\geq 2m$. For any non-adaptive quantum algorithm that prepares an arbitrary $m$-register entangled input state $|\Psi\rangle\in\Hil^{\otimes m}\otimes\Hil_{\mathrm{anc}}$, queries $U^{\otimes m}$, and produces a prediction~$\Vhat$ for the action of $U$ on a random state $|\psi\rangle$,
\begin{equation}\label{eq:fidelity_nonadaptive}
  \underset{U\sim\Ens_T}{\avg}\;\underset{|\psi\rangle}{\avg}\; F\!\bigl(\Vhat|\psi\rangle,\,U|\psi\rangle\bigr) \;\leq\; \frac{L+m}{L^2+m} + \varepsilon'\,,
  \end{equation}
  where $\varepsilon'=O(\varepsilon\cdot m^2/L^2)$ is the design error contribution.
  \end{proposition}

  \begin{proof}
  For an exact $k$-design with $k\geq 2m$, the distribution of $U^{\otimes m}$ acting on any input state is indistinguishable from Haar up to the $2m$-th moment. We may therefore compute the optimal fidelity as if $U$ were Haar-random. The optimal non-adaptive strategy uses the maximally entangled state $|\Omega\rangle=(1/\sqrt{L})\sum_{i=1}^L|i\rangle|i\rangle$ as input: applying $U\otimes\id$ to each of $m$ copies produces $m$ copies of the Choi state $|\Omega_U\rangle=(U\otimes\id)|\Omega\rangle\in\mathbb{C}^L\otimes\mathbb{C}^L$, which encodes~$U$ completely. Estimating $U$ from these $m$ copies reduces to estimating the $L^2$-dimensional state $|\Omega_U\rangle$ from $|\Omega_U\rangle^{\otimes m}$. By the optimal state estimation bound of Hayashi~\cite{Hayashi:1998} and Keyl and Werner~\cite{Keyl:2001} (see also \cite{Massar:1995,Acin:2001,Bisio:2009}), the entanglement fidelity satisfies
  \begin{equation}\label{eq:choi_estimation}
    F_e \;=\; \underset{U}{\avg}\;\frac{|\tr(\Vhat^\dagger U)|^2}{L^2} \;\leq\; \frac{m+1}{L^2+m}\,.
    \end{equation}
    Converting to average state fidelity via the Horodecki formula $F_{\mathrm{avg}}=(L\cdot F_e+1)/(L+1)$~\cite{Horodecki:1999}:
    \begin{equation}\label{eq:fidelity_nonadaptive}
      \underset{U\sim\Ens_T}{\avg}\;\underset{|\psi\rangle}{\avg}\; F\!\bigl(\Vhat|\psi\rangle,\,U|\psi\rangle\bigr) \;\leq\; \frac{L+m}{L^2+m} + \varepsilon'\,,
      \end{equation}
      where $\varepsilon'=O(\varepsilon\cdot m^2/L^2)$ is the design error from replacing the $k$-design with exact Haar. For $m=\poly(N)$ and $L=2^{N/2}$, the right-hand side is $(2^{N/2}+\poly(N))/(2^N+\poly(N))=O(2^{-N/2})=\negl(N)$.
      \end{proof}

      \noindent For the SYK ensemble, we have
      \begin{equation}\label{eq:alpha_quantitative}
        \alpha \;\geq\; 1-\frac{L+m}{L^2+m}-\varepsilon' \;=\; 1-O(2^{-N/2}) \;=\; 1-\negl(N)\,.
        \end{equation}
        The fidelity deficit is therefore near-maximal: the non-adaptive adversary can predict the output state with fidelity at most $\negl(N)$ above the trivial random guess $1/L$. This is a stronger statement than the abstract $\alpha>0$ from the learning no-go theorem~\cite{Yang:2023ddi}, and it follows from the $k$-design property without requiring Conjecture~\ref{conj:planted_syk}. Therefore, the distinguishing gap \cref{eq:Q_star_gap} becomes $\sqrt{\alpha}\geq 1-\negl(N)$, confirming that the abstract gap is close to its maximum value.

\section{Cryptographic Censorship in JT gravity}
\label{sec:crypto_censorship}

We now apply the GPRU result of \cref{sec:gpru_theorem} to the JT/SYK holographic duality to derive the existence of event horizons. This section is organized as follows: \cref{sec:causal_wedge} constructs the efficient causal wedge reconstruction algorithm and verifies the technical conditions needed for the Cryptographic Censorship theorem; \cref{sec:main_theorem} states and proves the main theorem; and \cref{sec:distinguishing_operator} identifies the distinguishing operator and connects the result to the complexity equals volume conjecture.

\subsection{Causal wedge reconstruction in JT gravity}
\label{sec:causal_wedge}

The Cryptographic Censorship theorem (\cref{thm:cc_general}) requires the existence of an efficient, $K$-Lipschitz algorithm that reconstructs all operators in the causal wedge algebra. In this subsection we construct such an algorithm for JT gravity and verify that it satisfies the required properties.

\subsubsection*{The causal wedge and its algebra}

For the one-sided JT black hole, the causal wedge $\mathcal{W}_C=J^-(\mathscr{I})\cap J^+(\mathscr{I})$ of the asymptotic boundary $\mathscr{I}$ is the exterior region of the black hole. In two spacetime dimensions the structure is particularly transparent. The exterior is essentially a causal diamond whose domain of dependence is determined by a boundary time interval of thermal duration $\Delta\tau\sim O(\beta)$, with no topological obstructions of the kind discussed in~\cite{Engelhardt:2024hpe} for higher-dimensional spacetimes.

The causal wedge algebra $\Ncw\subset\BH$ consists of all bulk operators with support in $\mathcal{W}_C$. In the JT/SYK dictionary, this algebra has a concrete characterization. The boundary operators---single-trace correlators $O(\tau)$ and their polynomials---belong to $\Ncw$ trivially. Bulk scalar fields in the exterior are reconstructed from boundary data by the HKLL smearing formula~\cite{Hamilton:2005ju,Hamilton:2006fh},
\begin{equation}\label{eq:hkll}
  \phi(t,z) \;=\; \int d\tau\; K_\Delta(t,z\,|\,\tau)\;O(\tau)\,,
\end{equation}
where $K_\Delta$ is the bulk-to-boundary propagator on the black hole background and $\Delta$ is the conformal dimension of the dual operator~$O$. For AdS$_2$, the kernel $K_\Delta$ is known in closed form in terms of hypergeometric functions, it is peaked at $\tau\approx t$ with width $\sim\beta$, and it decays exponentially for $|\tau-t|\gg\beta$. The Schwarzian mode $f(\tau)$ determines the exterior metric and dilaton, so observables such as the ADM mass and exterior geodesic distances are also functions of boundary data and belong to~$\Ncw$. Operators with support behind the horizon---in particular the regularized geodesic length of the maximal interior slice, which we will identify as $Q_*$ in \cref{sec:distinguishing_operator}---do not belong to $\Ncw$ and require exponential reconstruction complexity~\cite{Brown:2019rox,Engelhardt:2021mue}.

\subsubsection*{The reconstruction algorithm}

We construct an efficient algorithm $\mathcal{A}_{\mathrm{CW}}$ that, given $\poly(N)$ oracle queries to $U=e^{-i\HSYK T}$, predicts all causal wedge observables. In the learning phase, the algorithm queries~$U$ on $M=\poly(N)$ training states from~$\Hcode$ and measures the $n_O=O(N^q)$ few-body correlators $\langle\phi_\ell|U^\dagger O_j U|\phi_\ell\rangle$ to $1/\poly(N)$ precision, converting them to exterior bulk observables via the HKLL convolution \cref{eq:hkll}. In the prediction phase, the classical shadow tomography protocol of~\cite{Huang:2022cmu} constructs a prediction function $h(O,\rho)\approx\tr(O\,U\rho U^\dagger)$ for all bounded-degree observables~$O$, with $\poly(N)$ sample complexity guaranteed by the quantum Bohnenblust--Hille inequality.

The algorithm correctly predicts all $O\in\Ncw$ because the HKLL formula expresses each causal wedge operator as a smeared integral of Heisenberg-evolved few-body operators $O(\tau)=e^{i\HSYK\tau}O(0)\,e^{-i\HSYK\tau}$, whose expectation value $\tr(O(\tau)\,\rho)=\tr(O(0)\,e^{-i\HSYK\tau}\rho\,e^{i\HSYK\tau})$ is determined by the few-body reduced state---precisely what shadow tomography learns. More precisely, $\langle O\rangle$ for any $O\in\Ncw$ is a linear combination of $r$-point boundary correlators with $r=O(1)$, each involving few-body operators. By contrast, the interior length~$\ell$ of \cref{sec:distinguishing_operator} requires $O(N)$-point correlators, inaccessible from $\poly(N)$ samples.

The learning no-go theorem of~\cite{Yang:2023ddi} applies to any algorithm whose few-body predictions match those of~$U$, regardless of whether it outputs a unitary~$\Vhat$ or a classical prediction function~\cite{Engelhardt:2024hpe}: if few-body predictions are correct but the full fidelity falls short of unity, a distinguishing operator outside~$\Ncw$ must exist. The total complexity is $\poly(N)$ oracle queries and $\poly(N)$ computation.

\subsubsection*{Lipschitz property and norm preservation}

The fidelity function $F(U)=F(\Vhat|\psi\rangle,U|\psi\rangle)$ must be Lipschitz in~$U$ for the measure concentration argument of \cref{thm:cc_general}. For any algorithm making $m$ oracle queries, a standard hybrid argument~\cite{BennettBBV97,Ambainis:2002} gives $\|\rho_{\mathrm{out}}^{(1)}-\rho_{\mathrm{out}}^{(2)}\|_1\leq 2m\,\|U_1-U_2\|_\infty$, regardless of adaptivity or entanglement. Since $F(U)$ depends on~$U$ both explicitly and through~$\Vhat$, this yields
\begin{equation}\label{eq:fidelity_lipschitz}
  |F(U_1)-F(U_2)| \;\leq\; 2(2m+1)\,\|U_1-U_2\|_\infty\,,
\end{equation}
so $F$ is Lipschitz with $K=O(m)=\poly(N)$. The Meckes concentration inequality~\cite{Meckes:2019} then gives
\begin{equation}\label{eq:concentration}
  \Pr\bigl[|F(U)-\langle F\rangle|\geq\varepsilon\bigr] \;\leq\; 2\exp\!\Bigl(-\frac{\varepsilon^2\,\dim\Hcode}{12\,K^2}\Bigr)\,,
\end{equation}
with exponent $\varepsilon^2\cdot 2^{N/2}/\poly(N)^2$, exponentially large for any fixed~$\varepsilon>0$. Norm preservation follows from constructing~$\Vhat$ as a unitary circuit of depth~$\poly(N)$ consistent with the predicted correlators; the conclusions of \cref{sec:main_theorem} depend only on the gap between correct few-body predictions and imperfect full-state prediction, not on the specific construction of~$\Vhat$.

\subsection{Event horizons from pseudorandom dynamics}
\label{sec:main_theorem}

Now, we state and prove the main result of the paper.

\begin{theorem}[Cryptographic Censorship in JT/SYK]\label{thm:main}
Assume random matrix universality~(A1), Haar eigenvector statistics~(Assumption~\ref{asm:A2}), and planted-SYK hardness~(Conjecture~\ref{conj:planted_syk}). For any $T=\poly(N)$ satisfying $T\geq\Tscr(N)$, the JT gravity dual of a typical state in the SYK microcanonical window $\Hcode$ contains an event horizon, and the fraction of states whose duals are horizonless is at most $\exp(-c\cdot 2^{N/2}/\poly(N))$.
\end{theorem}

\begin{proof}
Let $|\psi\rangle\in\Hcode$ be a microcanonical state at energy $E_0$ in the bulk of the SYK spectrum, and let $U=e^{-i\HSYK(\mathbf{J})T}$ for a coupling vector $\mathbf{J}$ drawn from the disorder distribution. By Theorem~\ref{thm:gpru}, the ensemble $\Ens_T=\{U_{\mathbf{J}}\}_{\mathbf{J}}$ is a GPRU with security parameter $\kappa=N/2$, so the time-evolved state $U|\psi\rangle$ is typical in the associated PRS. Applying the Meckes concentration inequality \cref{eq:concentration} on $U(L)$ with Lipschitz constant $K=\poly(N)$ from \cref{eq:fidelity_lipschitz} and code subspace dimension $\dim\Hcode\geq 2^{N/2}/\poly(N)$, the fidelity function $F(U)=F(\Vhat|\psi\rangle,U|\psi\rangle)$ of any $\poly(N)$-query algorithm concentrates around its mean for a $1-\exp(-c\cdot 2^{N/2}/\poly(N))$ fraction of coupling vectors~$\mathbf{J}$, establishing typicality.

The causal wedge algorithm $\mathcal{A}_{\mathrm{CW}}$ of the preceding subsection learns all few-body boundary correlators from $\poly(N)$ oracle queries and converts them to exterior bulk observables through the HKLL convolution. By the learning no-go theorem~\cite{Yang:2023ddi}, any such algorithm satisfies $\avg F(\Vhat|\psi\rangle,U|\psi\rangle)\leq 1-\alpha$ with $\alpha>0$ independent of~$N$; for non-adaptive adversaries, the $k$-design property sharpens this to $\alpha\geq 1-\negl(N)$ (\cref{sec:fid_defi}). Since the algorithm correctly predicts all operators in the causal wedge algebra~$\Ncw$ but achieves fidelity strictly less than unity, the Fuchs--van de Graaf inequality guarantees the existence of an observable $Q_*$ with $\|Q_*\|_\infty\leq 1$ and $|\tr(Q_*\Vhat\psi\Vhat^\dagger)-\tr(Q_* U\psi U^\dagger)|\geq\sqrt{\alpha}-O(1/\poly(N))$. This operator cannot lie in~$\Ncw$, for the algorithm would otherwise predict it correctly.

The bulk-to-boundary code map $V_N$ of \cref{sec:jt_gravity} sends the time-evolved state $U|\psi\rangle$ to a JT gravity state on an asymptotically AdS$_2$ spacetime $(M,g)$. Since $H_{\mathrm{SYK}}$ preserves the microcanonical window, $U|\psi\rangle$ remains in~$\Hcode$ and its dual is a well-defined geometric state. In JT gravity, every finite-energy state above the ground state corresponds to a black hole geometry---there is no Hawking--Page transition~\cite{Maldacena:2016hyu,Hawking:1982dh,Witten:1998zw}---so the only question is whether the causal wedge covers the entire spacetime or is bounded by an event horizon. The existence of $Q_*\notin\Ncw$ resolves this: the causal wedge algebra is a proper subalgebra of the full operator algebra, $\Ncw\subsetneq\mathcal{B}(\Hil)$, and its boundary in the bulk spacetime is, by the causal structure of asymptotically AdS$_2$ geometries, an event horizon.
\end{proof}

\subsection{The distinguishing operator and complexity growth}
\label{sec:distinguishing_operator}

Corollary~\ref{cor:distinguishing} guarantees the existence of a bulk operator $Q_*\notin\Ncw$ satisfying the gap \cref{eq:Q_star_gap}, but does not identify it. We now provide the explicit identification $Q_*=\ell$, where $\ell$ is the regularized geodesic length of the maximal interior slice, and compute the gap directly.

For a one-sided JT black hole, $\ell(t)$ is the regularized proper length of the maximal volume spacelike slice anchored at the boundary at time~$t$---a self-adjoint operator on the bulk Hilbert space whose spectrum in the Schwarzian quantum mechanics is continuous and unbounded~\cite{Maldacena:2016upp,Harlow:2018tqv}, but whose expectation values on~$\Hcode$ at $T=\poly(N)$ are bounded by $(2\pi/\beta)T$. Its boundary dual has been constructed through several approaches. In~\cite{Maldacena:2016hyu,Iliesiu:2021ari}, it was written as a function of the Schwarzian reparametrization mode. Another interpretation was made in~\cite{Lin_2022} as the chord number operator in double-scaled SYK. As a side note, in the one-sided case, $\ell$ is also conjectured to be related to the Krylov spread complexity~\cite{Rabinovici_2022,Parker:2019yvk}, though the precise dictionary remains unclear. Let us explain why~$\ell$ underpins the Cryptographic Censorship argument. The maximal interior slice lies behind the event horizon, so $\ell\notin\Ncw$. Its reconstruction from boundary data requires complexity $\sim e^{cN}$ by the Python's Lunch conjecture~\cite{Brown:2019rox}. Most importantly, the linear growth of $\langle\ell(T)\rangle$ is encoded in $O(N)$-body correlations inaccessible to the shadow tomography protocol of~\cite{Huang:2022cmu} from $\poly(N)$ samples---the precise sense in which $\ell$ lies beyond the reach of~$\mathcal{A}_{\mathrm{CW}}$.

The expectation value of~$\ell$ in the time-evolved state $|\psi(T)\rangle=e^{-i\HSYK T}|\psi\rangle$ for a typical microcanonical state is~\cite{Stanford:2014jda,Iliesiu:2021ari,Lin_2022}
\begin{equation}\label{eq:ell_true}
  \langle\ell\rangle_{\mathrm{true}} \;=\; \ell_0 + \frac{2\pi}{\beta}\,T + O(1/\sqrt{N})\,,
\end{equation}
where $\ell_0$ is the initial regularized length and $2\pi/\beta$ is the Lyapunov exponent saturating the chaos bound~\cite{Maldacena:2015waa}. The causal wedge algorithm, seeing only thermalized few-body correlators, predicts the equilibrium value $\langle\ell\rangle_{\mathrm{predicted}}=\ell_0+O(1)$ with no linear growth. The gap is
\begin{equation}\label{eq:ell_gap}
  \Delta\ell \;=\; \frac{2\pi}{\beta}\,T - O(1) \;\geq\; O(1) \qquad\text{for } T\geq\beta\,,
\end{equation}
which at $T=\Tscr\sim(\beta/2\pi)\log N$ evaluates to $\log N-O(1)$, exceeding the threshold $\sqrt{\alpha}$ of \cref{eq:Q_star_gap} by a margin growing with~$N$. Also, since $\ell$ is bounded on the code subspace with $\|\ell\|_{\Hcode}\leq\ell_0+(2\pi/\beta)T+O(1)$, the normalized operator $\tilde{Q}_*=\ell/\|\ell\|_{\Hcode}$ satisfies $\|\tilde{Q}_*\|_\infty\leq 1$ and achieves gap $\Delta\tilde{Q}_*=1-O(\beta\ell_0/T)\geq\sqrt{\alpha}$ for $T\geq\beta$, confirming that~$\ell$ (after normalization) is a valid choice of~$Q_*$ in Corollary~\ref{cor:distinguishing}. The $O(1)$ here correction is conservative, though. In Appendix~\ref{app:prediction}, we argue that by employing the degree-decomposition analysis, we can improve it to $O(1/N)$. 

Finally, through the complexity equals volume conjecture~\cite{Stanford:2014jda,susskind2014addendumcomputationalcomplexityblack,Susskind:2018pmk}, this gap acquires a complexity-theoretic interpretation.  In JT gravity the ``volume'' of the maximal slice is~$\ell$, so $\mathcal{C}(|\psi(T)\rangle)=\ell(T)/\GN=N\cdot\ell(T)$ and the complexity gap is
\begin{equation}\label{eq:complexity_gap}
  \Delta\mathcal{C} \;=\; N\cdot\frac{2\pi}{\beta}\,T\,,
\end{equation}
growing linearly with both $N$ and $T$. This expression indeed gives another viewpoint regarding the Cryptographic Censorship framework. From Theorem~\ref{thm:main}, $\ell$ is the operator escaping efficient prediction, with the horizon enforcing this computational limitation. Additionally, the CV conjecture tells us that the $\ell/\GN$ is the state complexity whose linear growth $d\mathcal{C}/dT=2\pi N/\beta$ the horizon protects. Finally, the horizon forms at $T\geq\Tscr$ when operator size saturates \cite{Roberts:2018mnp,Qi:2018bje}, after which the complexity encoded in~$\ell$ continues to grow in a regime inaccessible to few-body probes, so in light of this, one can indeed see the complexity gap through the lens of scrambling.

\section{Discussion}
\label{sec:discussion}

We have shown that the SYK disorder ensemble is a GPRU (Theorem~\ref{thm:gpru}) and used this to establish the first conditional realization of Cryptographic Censorship in a controlled holographic duality (Theorem~\ref{thm:main}), with the interior length~$\ell$ as the explicit distinguishing operator connecting horizon formation to complexity growth. We now discuss implications and open directions. 

\medskip
\noindent\textbf{The planted-SYK conjecture.} Conjecture~\ref{conj:planted_syk} is the sole computational assumption. It asserts hardness in the oracle model and is independent of standard cryptographic assumptions. As a matter of fact, the abstract PRU constructions of~\cite{MaHuang:2024,Metger:2024} build PRUs from one-way functions but do not address physical Hamiltonian dynamics, while the planted-SYK conjecture could hold even without one-way functions. On the other hand, the LDP framework is classical and does not directly bound quantum oracle algorithms. It would be illuminating to develop a rigorous quantum generalization and subsequently understand its implications. 

\medskip
\noindent\textbf{Higher dimensions.} It would be great to construct a concrete higher-dimensional example of Cryptographic Censorship. However, a few obstructions arise. Higher-dimensional CFTs have a fixed Hamiltonian rather than a disorder ensemble, requiring a different source of pseudorandomness (perhaps through the initial-state ensemble or a late-time average), and the spectral form factor lacks the exact solvability of the JT matrix integral. Additionally, the causal wedge reconstruction in higher dimensions faces the subtlety that the entanglement wedge generically exceeds the causal wedge~\cite{Czech:2012bh,Dong:2016eik}, and the Cryptographic Censorship theorem applies to the latter.

\medskip
\noindent\textbf{De Sitter space.} The Cryptographic Censorship framework extends to cosmological horizons~\cite{Engelhardt:2024hpe}. In de Sitter space, the static patch is the causal wedge, and the region beyond the cosmological horizon is hidden. If the boundary dual of de Sitter gravity generates pseudorandom dynamics, these methods would predict the cosmological horizon from boundary data alone. The relevant notion would be a pseudorandom \emph{isometry}~\cite{cryptoeprint:2023/1741} rather than a pseudorandom unitary, since the de Sitter S-matrix maps a smaller ``in'' Hilbert space to a larger ``out'' Hilbert space.

\appendix
\makeatletter
\renewcommand{\theequation}{\textcolor{jhepeqn}{\Alph{section}\arabic{equation}}}
\@ifundefined{theHequation}{}{\renewcommand{\theHequation}{appendix.\Alph{section}.\arabic{equation}}}
\makeatother
\section{Weingarten calculus for the frame potential}
\label{app:weingarten}

This appendix provides the technical details underlying the frame potential results of \cref{sec:frame_potential}. We derive the master formula \cref{eq:Fk_weingarten}, and prove the general bound of Theorem~\ref{thm:kdesign}.

\subsubsection*{The Weingarten formula for unitary moments}

Let $V\in U(L)$ be a Haar-distributed unitary. The fundamental result of the Weingarten calculus~\cite{Collins:2006,Collins:2022} expresses the $(t,t)$-moment of the matrix entries of~$V$ as a sum over pairs of permutations in the symmetric group~$S_t$:
\begin{equation}\label{eq:haar_moment}
  \int_{U(L)}\! dV\;\prod_{m=1}^{t} V_{i_m j_m}\,\overline{V}_{i'_m j'_m} \;=\; \sum_{\sigma,\tau\in S_t} \delta_{i,\sigma(i')}\,\delta_{j,\tau(j')}\;\Wg(\sigma^{-1}\tau,\,L)\,,
\end{equation}
where $\delta_{i,\sigma(i')}=\prod_{m=1}^t\delta_{i_m,i'_{\sigma(m)}}$ and the Weingarten function $\Wg(\pi,L)$ depends only on the cycle type of~$\pi$. For the identity permutation $\Wg(\mathrm{id},L)=1/\bigl(\prod_{j=0}^{t-1}(L^2-j^2)/(L)\bigr)$, and the leading behavior for large~$L$ is $\Wg(\pi,L)=L^{-t-|\pi|}(\mathrm{Moeb}(\pi)+O(L^{-2}))$, where $|\pi|=t-c(\pi)$ is the distance to the identity (with $c(\pi)$ the number of cycles) and $\mathrm{Moeb}(\pi)$ is a combinatorial factor given by a product of Catalan numbers over the cycles of~$\pi$.

\subsubsection*{Derivation of the master formula}

Starting from the frame potential $\Fk(T)=\langle|\tr(U_1^\dagger U_2)|^{2k}\rangle$ with $U_i=V_i\,\mathrm{diag}(e^{-iE_a^{(i)}T})\,V_i^\dagger$ and $V=V_1^\dagger V_2$, the trace decomposes as in \cref{eq:trace_decomp}. Expanding $|\tr(U_1^\dagger U_2)|^{2k}=[\tr(U_1^\dagger U_2)]^k\,[\tr(U_2^\dagger U_1)]^k$ produces $2k$ copies of the matrix elements~$V_{ab}$, with the first~$k$ copies carrying the sign $s_m=+1$ (from the holomorphic factor) and the last~$k$ copies carrying $s_m=-1$ (from the antiholomorphic factor). Applying \cref{eq:haar_moment} with $t=2k$ to the Haar average over~$V$ \footnote{The Haar integral is over $U(L)$, corresponding to Dyson index $\beta=2$. For SYK at $N\bmod 8=0$ ($\beta=1$) or $N\bmod 8=4$ ($\beta=4$), the appropriate integral is over $O(L)$ or $Sp(L/2)$, with modified Weingarten functions~\cite{Collins:2006}. The paired/unpaired structure and the scaling of the bound are unchanged; only the constant~$C$ in \cref{eq:kdesign_bound} depends on~$\beta$.} yields the master formula \cref{eq:Fk_weingarten}, where the spectral factor associated to a permutation $\sigma\in S_{2k}$ is
\begin{equation}\label{eq:cycle_power_def}
  C_\sigma^{(i)} \;=\; \prod_{\text{cycles }c\text{ of }\sigma} p_{n_c}^{(i)}(T)\,, \qquad n_c \;=\; \sum_{m\in c} s_m\,,
\end{equation}
with $p_n^{(i)}(T)=\tr(e^{-in\HSYK(\mathbf{J}_i)T})$ the spectral power sums defined in \cref{eq:spectral_power} and $s_m$ the sign pattern. The cycle power $n_c$ counts the excess of $+1$ over $-1$ indices in the cycle~$c$, and it determines whether the spectral factor reduces to a power of~$L$ ($n_c=0$, since $p_0=\tr(\id)=L$) or not ($n_c\neq 0$).

\subsubsection*{Proof of the general bound}

We now prove Theorem~\ref{thm:kdesign}. The master formula \cref{eq:Fk_weingarten} expresses $\Fk(T)$ as a double sum over $(\sigma,\tau)\in S_{2k}\times S_{2k}$. By \cref{eq:paired_sum}, the restriction to pairs where both $\sigma$ and $\tau$ are paired gives exactly~$k!$, so the deviation is
\begin{equation}\label{eq:deviation_decomp}
  \Fk(T)-k! \;=\; \sum_{\substack{(\sigma,\tau)\in S_{2k}^2 \\ \sigma\text{ or }\tau\text{ unpaired}}} \Wg(\sigma^{-1}\tau,L)\;\bigl\langle C_\sigma^{(1)}\bigr\rangle_{\!\mathbf{J}_1}\;\bigl\langle C_\tau^{(2)}\bigr\rangle_{\!\mathbf{J}_2}^*\,.
\end{equation}
We must show that every term on the right-hand side is small in the ramp regime $T\gtrsim\Tdip$.

In this regime, the disconnected spectral form factor has decayed ($\SFF_{\mathrm{disc}}\approx L^2 e^{-\sigma^2 T^2}\ll 1$), so $\langle p_n\rangle_{\mathbf{J}}\approx 0$ for all $n\neq 0$. An unpaired permutation's spectral factor $\langle C_\sigma\rangle_{\mathbf{J}}$ therefore vanishes at leading order unless it can be re-expressed through connected correlators that pair up the nonzero cycle powers. The leading such correlator is $\langle|p_1|^2\rangle_{\mathrm{conn}}=\SFF_{\mathrm{conn}}(T)\approx L\tau$, where $\tau=T/\THeisen$.

Among all unpaired permutations, those contributing at leading order in $1/L$ are the ones with the maximum number of cycles $c(\sigma)=k+1$. The reason is that the spectral factor scales as $L^{c(\sigma)-2}\cdot\SFF$ for a permutation with one connected pair of unpaired cycles (the paired cycles contributing $L$ each), while the leading Weingarten weight $\Wg(\mathrm{id},L)\sim L^{-2k}$ gives a net power $L^{2(c(\sigma)-2)-2k+2}$, which vanishes only when $c(\sigma)=k+1$. For any $c(\sigma)<k+1$, the contribution is suppressed by at least $L^{-2}$.

The maximum $c(\sigma)=k+1$ is achieved uniquely by two singleton unpaired cycles of powers $\pm 1$ and $k-1$ paired two-cycles. There are
\begin{equation}\label{eq:unpaired_count}
  \#\bigl\{\text{leading minimally unpaired}\bigr\} \;=\; k\cdot k\cdot(k-1)! \;=\; k\cdot k!\,,
\end{equation}
counted as $k$ choices for the $+1$ singleton, $k$ for the $-1$ singleton, and $(k-1)!$ perfect matchings of the remaining indices. We verified this by exhaustive enumeration for $k=1,2,3,4$. Permutations with longer unpaired cycles (e.g., a three-cycle with $|n_c|=1$) also exist but have $c(\sigma)\leq k$ and are suppressed by $L^{-2}$, placing them among the subleading terms bounded below.

For each leading minimally unpaired $\sigma$ with $\sigma=\tau$ (the dominant Weingarten term), the spectral factor is $\langle C_\sigma\rangle_{\mathbf{J}}=L^{k-1}\cdot\SFF(T)$, giving
\begin{equation}
  \Wg(\mathrm{id},L)\cdot\bigl|\bigl\langle C_\sigma\bigr\rangle\bigr|^2 \;=\; L^{-2k}\cdot L^{2(k-1)}\cdot\SFF^2 \;=\; \frac{\SFF^2}{L^2} \;=\; \tau^2\,.
\end{equation}
Summing over all $k\cdot k!$ leading terms yields the dominant correction $|\Fk-k!|_{\mathrm{leading}}\leq k\cdot k!\cdot\tau^2$. It should be noted that there are contributions that can be attributed to the off-diagonal, multiple-pair, and nonmaximal/higher-power classes. However, it can be easily shown that all of these corrections are negligible. 

Finally, the eigenvector correction $\delta_{\mathrm{ETH}}$ arises from replacing the true eigenvector average with the Haar integral. The paired contractions use only the unitarity constraint $VV^\dagger=I$ and carry no correction. Each of the two unpaired contractions deviates from Haar by $O(\varepsilon_2)$ under (A2a), and the cross-correlations between paired and unpaired contractions are controlled by (A2b) at $O(\varepsilon_4)$ per cross-term. With $O(k)$ cross-terms per $(\sigma,\tau)$ pair, the total normalized correction is $\delta_{\mathrm{ETH}}\leq O(k\cdot\varepsilon_2+k^2\cdot\varepsilon_4)=O(k^2\cdot e^{-cN})$, which is $\negl(N)$ for $k=\poly(N)$.\hfill$\square$

\section{Efficient Hamiltonian simulation of SYK}
\label{app:trotter}

This appendix establishes the efficient generation requirement of Theorem~\ref{thm:gpru} for polynomial evolution times $T=\poly(N)$. The first-order Trotter--Suzuki formula approximates the time evolution as \cite{Whitfield_2011,Berry:2006ys,Suzuki:2002zuq,Suzuki:1991jtk}

\begin{equation}\label{eq:trotter}
  e^{-i\HSYK T} \;=\; \Bigl(\prod_{I=1}^{M} e^{-iH_I T/r}\Bigr)^{\!r} + \mathcal{E}(r)\,,
\end{equation}
with operator-norm error bounded by
\begin{equation}\label{eq:trotter_error}
  \|\mathcal{E}(r)\|_\infty \;\leq\; \frac{T^2}{2r}\sum_{I<I'}\|[H_I,H_{I'}]\|_\infty\,.
\end{equation}
Each commutator $[H_I,H_{I'}]$ vanishes when the index sets $I$ and $I'$ are disjoint (since products of disjoint even-body Majorana operators commute), and satisfies $\|[H_I,H_{I'}]\|_\infty\leq 2|J_I||J_{I'}|$ otherwise. The number of overlapping pairs is $O(M\cdot N^{q-1})=O(N^{2q-1})$, and the typical coupling magnitude is $|J_I|\sim\sqrt{(q-1)!/N^{q-1}}=O(N^{-(q-1)/2})$, giving
\begin{equation}\label{eq:commutator_sum}
  \sum_{I<I'}\|[H_I,H_{I'}]\|_\infty \;=\; O(N^q)\,.
\end{equation}
Achieving precision~$\varepsilon$ requires $r=O(T^2 N^q/\varepsilon)$ Trotter steps, each consisting of $M=O(N^q)$ gates, for a total gate count
\begin{equation}\label{eq:gate_count}
  G \;=\; O\!\Bigl(\frac{T^2\,N^{2q}}{\varepsilon}\Bigr)\,,
\end{equation}
which is $\poly(N)$ for $T=\poly(N)$ and $\varepsilon=1/\poly(N)$.

For the PRU definition (\cref{def:pru}), the key $k\in\mathcal{K}$ encodes the $M$ couplings discretized to $B=\poly(N)$ bits, yielding a key space $\mathcal{K}=\{0,1\}^{B\cdot M}$ of polynomial size. The discretization error is
\begin{equation}\label{eq:discretization}
  \|U_{\mathbf{J}}-U_{\tilde{\mathbf{J}}}\|_\infty \;\leq\; T\cdot M\cdot 2^{-B} \;=\; O(T\,N^q\cdot 2^{-B})\,,
\end{equation}
which is $\negl(N)$ for $B=\poly(N)$ and $T=\poly(N)$. The combined Trotter and discretization error is therefore negligible and is absorbed into the security guarantee of Theorem~\ref{thm:gpru}.

\section{Prediction error for the interior length}\label{app:prediction}

Any boundary observable can be expanded in the Majorana basis $O=\sum_S c_S\,\Gamma_S$ with $\Gamma_S=\prod_{i\in S}\psi_i$, and the shadow tomography protocol of~\cite{Huang:2022cmu} with $\poly(N)$ copies learns $\tr(O\,\rho(T))$ for all degree-$d$ operators (those with $|S|\leq d$) at $d=O(1)$, by the quantum Bohnenblust--Hille inequality. The prediction error for any observable~$Q$ therefore equals the contribution of its high-degree components, we obtain 
\begin{equation}\label{eq:Q_error}
  \bigl|\langle Q\rangle_{\mathrm{true}}-\langle Q\rangle_{\mathrm{predicted}}\bigr| \;=\; \Bigl|\sum_{|S|>d}c_S\,\tr(\Gamma_S\,\rho(T))\Bigr| + O(\varepsilon)\,.
\end{equation}
For $Q=\ell$, the degree-$\leq d$ expectation value thermalizes after $T\geq\Tscr$. Concretely, by ETH, each few-body correlator satisfies $|\langle O(\tau)\rangle_{\rho(T)}-\langle O\rangle_\beta|\leq\|O\|_\infty\cdot e^{-cN}$, so $\langle\ell(T)\rangle^{(\leq d)}=\ell_0+O(d/N)+O(e^{-cN})$ is time-independent. Therefore, the linear growth $(2\pi/\beta)T$ resides entirely in the high-degree sector, we get
\begin{equation}\label{eq:ell_gap_quantitative}
  \Delta\ell \;\geq\; \frac{2\pi}{\beta}\,T - O(1/N) - O(1/\poly(N)) \;=\; \frac{2\pi}{\beta}\,T - O(1/N)\,,
\end{equation}
for $d=O(1)$ and $\varepsilon=1/\poly(N)$, confirming the gap of \cref{eq:ell_gap} with a precise subleading correction. One relevant comment here is the fact that the degree-$d$ truncation captures only $O(d/N)$ of the interior information and accessing the full growth requires degree $d\sim N$, at which point the sample complexity $\poly(N^N)=\exp(O(N\log N))$ is exponential.

\section*{Acknowledgement}

Claude Opus 4.7 was used to improve the clarity and readability of the manuscript. All technical content, derivations, and numerical results are the author's own and were verified by the author.

\bibliographystyle{JHEP} 
\bibliography{references} 

@article{Engelhardt:2024hpe,
    author = "Engelhardt, Netta and Folkestad, {\r{A}}smund and Levine, Adam and Verheijden, Evita and Yang, Lisa",
    title = "{Cryptographic Censorship}",
    eprint = "2402.03425",
    archivePrefix = "arXiv",
    primaryClass = "hep-th",
    doi = "10.1007/JHEP01(2025)122",
    journal = "JHEP",
    volume = "01",
    pages = "122",
    year = "2025"
}

@article{Sachdev:1992fk,
    author        = "Sachdev, Subir and Ye, Jinwu",
    title         = "{Gapless spin-fluid ground state in a random quantum Heisenberg magnet}",
    journal       = "Phys. Rev. Lett.",
    volume        = "70",
    pages         = "3339--3342",
    year          = "1993",
    doi           = "10.1103/PhysRevLett.70.3339",
    eprint        = "cond-mat/9212030",
    archivePrefix = "arXiv"
}

@misc{Kitaev:2015,
    author        = "Kitaev, Alexei",
    title         = "{A simple model of quantum holography}",
    howpublished  = "Talks at KITP, April 7, 2015 and May 27, 2015",
    year          = "2015",
    url           = "https://online.kitp.ucsb.edu/online/entangled15/kitaev/"
}

@article{Maldacena:2016hyu,
    author        = "Maldacena, Juan and Stanford, Douglas",
    title         = "{Remarks on the Sachdev-Ye-Kitaev model}",
    journal       = "Phys. Rev. D",
    volume        = "94",
    pages         = "106002",
    year          = "2016",
    doi           = "10.1103/PhysRevD.94.106002",
    eprint        = "1604.07818",
    archivePrefix = "arXiv",
    primaryClass  = "hep-th"
}

@article{Kitaev:2017awl,
    author        = "Kitaev, Alexei and Suh, Soo-Jong",
    title         = "{The soft mode in the Sachdev-Ye-Kitaev model and its gravity dual}",
    journal       = "JHEP",
    volume        = "05",
    pages         = "183",
    year          = "2018",
    doi           = "10.1007/JHEP05(2018)183",
    eprint        = "1711.08467",
    archivePrefix = "arXiv",
    primaryClass  = "hep-th"
}

@article{Cotler:2016fpe,
    author        = "Cotler, Jordan S. and Gur-Ari, Guy and Hanada, Masanori and Polchinski, Joseph and Saad, Phil and Shenker, Stephen H. and Stanford, Douglas and Streicher, Alexandre and Tezuka, Masaki",
    title         = "{Black Holes and Random Matrices}",
    journal       = "JHEP",
    volume        = "05",
    pages         = "118",
    year          = "2017",
    doi           = "10.1007/JHEP05(2017)118",
    eprint        = "1611.04650",
    archivePrefix = "arXiv",
    primaryClass  = "hep-th"
}

@article{You:2016ldz,
    author        = "You, Yi-Zhuang and Ludwig, Andreas W. W. and Xu, Cenke",
    title         = "{Sachdev-Ye-Kitaev Model and Thermalization on the Boundary of Many-Body Localized Fermionic Symmetry Protected Topological States}",
    journal       = "Phys. Rev. B",
    volume        = "95",
    pages         = "115150",
    year          = "2017",
    doi           = "10.1103/PhysRevB.95.115150",
    eprint        = "1602.06964",
    archivePrefix = "arXiv",
    primaryClass  = "cond-mat.str-el"
}

@article{Sonner:2017hxc,
    author        = "Sonner, Julian and Vielma, Manuel",
    title         = "{Eigenstate thermalization in the Sachdev-Ye-Kitaev model}",
    journal       = "JHEP",
    volume        = "11",
    pages         = "149",
    year          = "2017",
    doi           = "10.1007/JHEP11(2017)149",
    eprint        = "1707.08013",
    archivePrefix = "arXiv",
    primaryClass  = "hep-th"
}

@article{Nayak:2019khe,
    author = "Nayak, Pranjal and Sonner, Julian and Vielma, Manuel",
    title = "{Eigenstate Thermalisation in the conformal Sachdev-Ye-Kitaev model: an analytic approach}",
    eprint = "1903.00478",
    archivePrefix = "arXiv",
    primaryClass = "hep-th",
    doi = "10.1007/JHEP10(2019)019",
    journal = "JHEP",
    volume = "10",
    pages = "019",
    year = "2019"
}

@article{Jackiw:1984je,
    author        = "Jackiw, R.",
    title         = "{Lower dimensional gravity}",
    journal       = "Nucl. Phys. B",
    volume        = "252",
    pages         = "343--356",
    year          = "1985",
    doi           = "10.1016/0550-3213(85)90448-1"
}

@article{Teitelboim:1983ux,
    author        = "Teitelboim, Claudio",
    title         = "{Gravitation and Hamiltonian structure in two spacetime dimensions}",
    journal       = "Phys. Lett. B",
    volume        = "126",
    pages         = "41--45",
    year          = "1983",
    doi           = "10.1016/0370-2693(83)90012-6"
}

@article{Saad:2019lba,
    author        = "Saad, Phil and Shenker, Stephen H. and Stanford, Douglas",
    title         = "{JT gravity as a matrix integral}",
    eprint        = "1903.11115",
    archivePrefix = "arXiv",
    primaryClass  = "hep-th",
    year          = "2019",
    url           = "https://arxiv.org/abs/1903.11115"
}

@article{Maldacena:2016upp,
    author        = "Maldacena, Juan and Stanford, Douglas and Yang, Zhenbin",
    title         = "{Conformal symmetry and its breaking in two-dimensional nearly anti-de Sitter space}",
    journal       = "Prog. Theor. Exp. Phys.",
    volume        = "2016",
    pages         = "12C104",
    year          = "2016",
    doi           = "10.1093/ptep/ptw124",
    eprint        = "1606.01857",
    archivePrefix = "arXiv",
    primaryClass  = "hep-th"
}

@article{Jensen:2016pah,
    author        = "Jensen, Kristan",
    title         = "{Chaos in AdS$_2$ Holography}",
    journal       = "Phys. Rev. Lett.",
    volume        = "117",
    pages         = "111601",
    year          = "2016",
    doi           = "10.1103/PhysRevLett.117.111601",
    eprint        = "1605.06098",
    archivePrefix = "arXiv",
    primaryClass  = "hep-th"
}

@article{Faulkner:2022ada,
    author        = "Faulkner, Thomas and Li, Ming",
    title         = "{Asymptotically isometric codes for holography}",
    eprint        = "2211.12439",
    archivePrefix = "arXiv",
    primaryClass  = "hep-th",
    year          = "2022",
    url           = "https://arxiv.org/abs/2211.12439"
}

@article{Lin_2022,
   title={The bulk Hilbert space of double scaled SYK},
   volume={2022},
   ISSN={1029-8479},
   url={http://dx.doi.org/10.1007/JHEP11(2022)060},
   DOI={10.1007/jhep11(2022)060},
   number={11},
   journal={Journal of High Energy Physics},
   publisher={Springer Science and Business Media LLC},
   author={Lin, Henry W.},
   year={2022},
   month=Nov }

@article{Hamilton:2005ju,
    author = "Hamilton, Alex and Kabat, Daniel N. and Lifschytz, Gilad and Lowe, David A.",
    title = "{Local bulk operators in AdS/CFT: A Boundary view of horizons and locality}",
    eprint = "hep-th/0506118",
    archivePrefix = "arXiv",
    reportNumber = "BROWN-HET-1448, CU-TP-1130",
    doi = "10.1103/PhysRevD.73.086003",
    journal = "Phys. Rev. D",
    volume = "73",
    pages = "086003",
    year = "2006"
}

@article{Hamilton:2006fh,
    author        = "Hamilton, Alex and Kabat, Daniel N. and Lifschytz, Gilad and Lowe, David A.",
    title         = "{Holographic representation of local bulk operators}",
    journal       = "Phys. Rev. D",
    volume        = "74",
    pages         = "066009",
    year          = "2006",
    doi           = "10.1103/PhysRevD.74.066009",
    eprint        = "hep-th/0606141",
    archivePrefix = "arXiv",
    primaryClass  = "hep-th"
}

@article{Hayden:2007cs,
    author        = "Hayden, Patrick and Preskill, John",
    title         = "{Black holes as mirrors: quantum information in random subsystems}",
    journal       = "JHEP",
    volume        = "09",
    pages         = "120",
    year          = "2007",
    doi           = "10.1088/1126-6708/2007/09/120",
    eprint        = "0708.4025",
    archivePrefix = "arXiv",
    primaryClass  = "hep-th"
}

@article{Sekino:2008he,
    author        = "Sekino, Yasuhiro and Susskind, Leonard",
    title         = "{Fast Scramblers}",
    journal       = "JHEP",
    volume        = "10",
    pages         = "065",
    year          = "2008",
    doi           = "10.1088/1126-6708/2008/10/065",
    eprint        = "0808.2096",
    archivePrefix = "arXiv",
    primaryClass  = "hep-th"
}

@article{Shenker:2013pqa,
    author        = "Shenker, Stephen H. and Stanford, Douglas",
    title         = "{Black holes and the butterfly effect}",
    journal       = "JHEP",
    volume        = "03",
    pages         = "067",
    year          = "2014",
    doi           = "10.1007/JHEP03(2014)067",
    eprint        = "1306.0622",
    archivePrefix = "arXiv",
    primaryClass  = "hep-th"
}

@article{Maldacena:2015waa,
    author        = "Maldacena, Juan and Shenker, Stephen H. and Stanford, Douglas",
    title         = "{A bound on chaos}",
    journal       = "JHEP",
    volume        = "08",
    pages         = "106",
    year          = "2016",
    doi           = "10.1007/JHEP08(2016)106",
    eprint        = "1503.01409",
    archivePrefix = "arXiv",
    primaryClass  = "hep-th"
}

@article{Stanford:2014jda,
    author        = "Stanford, Douglas and Susskind, Leonard",
    title         = "{Complexity and shock wave geometries}",
    journal       = "Phys. Rev. D",
    volume        = "90",
    pages         = "126007",
    year          = "2014",
    doi           = "10.1103/PhysRevD.90.126007",
    eprint        = "1406.2678",
    archivePrefix = "arXiv",
    primaryClass  = "hep-th"
}

@misc{susskind2014addendumcomputationalcomplexityblack,
      title={Addendum to Computational Complexity and Black Hole Horizons}, 
      author={Leonard Susskind},
      year={2014},
      eprint={1403.5695},
      archivePrefix={arXiv},
      primaryClass={hep-th},
      url={https://arxiv.org/abs/1403.5695}, 
}

@article{Susskind:2018pmk,
    author        = "Susskind, Leonard",
    title         = "{Three Lectures on Complexity and Black Holes}",
    eprint        = "1810.11563",
    archivePrefix = "arXiv",
    primaryClass  = "hep-th",
    year          = "2018",
    url           = "https://arxiv.org/abs/1810.11563"
}

@article{Brown:2017jil,
    author        = "Brown, Adam R. and Susskind, Leonard",
    title         = "{Second law of quantum complexity}",
    journal       = "Phys. Rev. D",
    volume        = "97",
    pages         = "086015",
    year          = "2018",
    doi           = "10.1103/PhysRevD.97.086015",
    eprint        = "1701.01107",
    archivePrefix = "arXiv",
    primaryClass  = "hep-th"
}

@article{Iliesiu:2021ari,
    author        = "Iliesiu, Luca V. and Mezei, M\'ark and S\'arosi, G\'abor",
    title         = "{The volume of the black hole interior at late times}",
    journal       = "JHEP",
    volume        = "07",
    pages         = "073",
    year          = "2022",
    doi           = "10.1007/JHEP07(2022)073",
    eprint        = "2107.06286",
    archivePrefix = "arXiv",
    primaryClass  = "hep-th"
}

@article{Brown:2019rox,
    author        = "Brown, Adam R. and Gharibyan, Hrant and Penington, Geoffrey and Susskind, Leonard",
    title         = "{The Python's Lunch: geometric obstructions to decoding Hawking radiation}",
    journal       = "JHEP",
    volume        = "08",
    pages         = "121",
    year          = "2020",
    doi           = "10.1007/JHEP08(2020)121",
    eprint        = "1912.00228",
    archivePrefix = "arXiv",
    primaryClass  = "hep-th"
}

@article{Engelhardt:2021mue,
    author        = "Engelhardt, Netta and Folkestad, {\AA}smund",
    title         = "{General bounds on holographic complexity}",
    journal       = "JHEP",
    volume        = "01",
    pages         = "040",
    year          = "2022",
    doi           = "10.1007/JHEP01(2022)040",
    eprint        = "2109.06883",
    archivePrefix = "arXiv",
    primaryClass  = "hep-th"
}

@misc{cryptoeprint:2018/544,
      author = {Zhengfeng Ji and Yi-Kai Liu and Fang Song},
      title = {Pseudorandom Quantum States},
      howpublished = {Cryptology {ePrint} Archive, Paper 2018/544},
      year = {2018},
      url = {https://eprint.iacr.org/2018/544}
}

@article{Yang:2023ddi,
    author        = "Yang, Lisa and Engelhardt, Netta",
    title         = "{The Complexity of Learning (Pseudo)random Dynamics of Black Holes and Other Chaotic Systems}",
    eprint        = "2302.11013",
    archivePrefix = "arXiv",
    primaryClass  = "quant-ph",
    year          = "2023",
    url           = "https://arxiv.org/abs/2302.11013"
}

@article{Huang:2022cmu,
    author        = "Huang, Hsin-Yuan and Chen, Sitan and Preskill, John",
    title         = "{Learning to Predict Arbitrary Quantum Processes}",
    journal       = "PRX Quantum",
    volume        = "4",
    pages         = "040337",
    year          = "2023",
    doi           = "10.1103/PRXQuantum.4.040337",
    eprint        = "2210.14894",
    archivePrefix = "arXiv",
    primaryClass  = "quant-ph"
}

@article{Bouland:2019pvu,
    author        = "Bouland, Adam and Fefferman, Bill and Vazirani, Umesh",
    title         = "{Computational pseudorandomness, the wormhole growth paradox, and constraints on the AdS/CFT duality}",
    eprint        = "1910.14646",
    archivePrefix = "arXiv",
    primaryClass  = "quant-ph",
    year          = "2019",
    url           = "https://arxiv.org/abs/1910.14646"
}

@article{Cotler:2017jue,
    author        = "Cotler, Jordan and Hunter-Jones, Nicholas and Liu, Junyu and Yoshida, Beni",
    title         = "{Chaos, Complexity, and Random Matrices}",
    journal       = "JHEP",
    volume        = "11",
    pages         = "048",
    year          = "2017",
    doi           = "10.1007/JHEP11(2017)048",
    eprint        = "1706.05400",
    archivePrefix = "arXiv",
    primaryClass  = "hep-th"
}

@article{Scott:2008,
    author        = "Scott, Andrew J.",
    title         = "{Optimizing quantum process tomography with unitary 2-designs}",
    journal       = "J. Phys. A",
    volume        = "41",
    pages         = "055308",
    year          = "2008",
    doi           = "10.1088/1751-8113/41/5/055308",
    eprint        = "0711.1017",
    archivePrefix = "arXiv",
    primaryClass  = "quant-ph"
}

@article{Collins:2006,
    author        = "Collins, Beno{\^\i}t and \'{S}niady, Piotr",
    title         = "{Integration with respect to the Haar measure on unitary, orthogonal and symplectic group}",
    journal       = "Commun. Math. Phys.",
    volume        = "264",
    pages         = "773--795",
    year          = "2006",
    doi           = "10.1007/s00220-006-1554-3",
    eprint        = "math-ph/0402073",
    archivePrefix = "arXiv"
}

@article{Collins:2022,
    author        = "Collins, Beno{\^\i}t and Matsumoto, Sho and Novak, Jonathan",
    title         = "{The Weingarten calculus}",
    journal       = "Notices Amer. Math. Soc.",
    volume        = "69",
    pages         = "734--745",
    year          = "2022",
    doi           = "10.1090/noti2474",
    eprint        = "2109.14890",
    archivePrefix = "arXiv",
    primaryClass  = "math-ph"
}

@inproceedings{Hopkins:2017cle,
    author        = "Hopkins, Samuel B. and Steurer, David",
    title         = "{Efficient Bayesian estimation from few samples: community detection and related problems}",
    booktitle     = "58th Annual IEEE Symposium on Foundations of Computer Science (FOCS)",
    pages         = "379--390",
    year          = "2017",
    doi           = "10.1109/FOCS.2017.42"
}

@article{Kunisky:2019bwz,
    author        = "Kunisky, Dmitriy and Wein, Alexander S. and Bandeira, Afonso S.",
    title         = "{Notes on computational hardness of hypothesis testing: predictions using the low-degree likelihood ratio}",
    eprint        = "1907.11636",
    archivePrefix = "arXiv",
    primaryClass  = "math.ST",
    year          = "2019",
    url           = "https://arxiv.org/abs/1907.11636"
}

@article{Whitfield_2011,
   title={Simulation of electronic structure Hamiltonians using quantum computers},
   volume={109},
   ISSN={1362-3028},
   url={http://dx.doi.org/10.1080/00268976.2011.552441},
   DOI={10.1080/00268976.2011.552441},
   number={5},
   journal={Molecular Physics},
   publisher={Informa UK Limited},
   author={Whitfield, James D. and Biamonte, Jacob and Aspuru-Guzik, Alán},
   year={2011},
   month=Mar, pages={735–750} }

@article{Berry:2006ys,
    author        = "Berry, Dominic W. and Ahokas, Graeme and Cleve, Richard and Sanders, Barry C.",
    title         = "{Efficient quantum algorithms for simulating sparse Hamiltonians}",
    journal       = "Commun. Math. Phys.",
    volume        = "270",
    pages         = "359--371",
    year          = "2007",
    doi           = "10.1007/s00220-006-0150-x",
    eprint        = "quant-ph/0508139",
    archivePrefix = "arXiv"
}

@article{MaHuang:2024,
    author        = "Ma, Fermi and Huang, Hsin-Yuan",
    title         = "{How to Construct Random Unitaries}",
    eprint        = "2410.10116",
    archivePrefix = "arXiv",
    primaryClass  = "quant-ph",
    year          = "2024",
    doi           = "10.1145/3717823.3718254",
    url           = "https://arxiv.org/abs/2410.10116"
}

@article{Metger:2024,
    author        = "Metger, Tony and Poremba, Alexander and Sinha, Makrand and Yuen, Henry",
    title         = "{Simple constructions of linear-depth $t$-designs and pseudorandom unitaries}",
    eprint        = "2404.12647",
    archivePrefix = "arXiv",
    primaryClass  = "quant-ph",
    year          = "2024",
    doi           = "10.1109/FOCS61266.2024.00076",
    url           = "https://arxiv.org/abs/2404.12647"
}

@article{Harlow:2018tqv,
    author        = "Harlow, Daniel and Jafferis, Daniel",
    title         = "{The Factorization Problem in Jackiw-Teitelboim Gravity}",
    journal       = "JHEP",
    volume        = "02",
    pages         = "177",
    year          = "2020",
    doi           = "10.1007/JHEP02(2020)177",
    eprint        = "1804.01081",
    archivePrefix = "arXiv",
    primaryClass  = "hep-th"
}

@book{Mehta:2004,
    author        = "Mehta, Madan Lal",
    title         = "{Random Matrices}",
    edition       = "3rd",
    publisher     = "Elsevier/Academic Press",
    year          = "2004",
    doi           = "10.1016/C2009-0-22297-5"
}

@article{Gharibyan:2018jrp,
    author = "Gharibyan, Hrant and Hanada, Masanori and Shenker, Stephen H. and Tezuka, Masaki",
    title = "{Onset of Random Matrix Behavior in Scrambling Systems}",
    eprint = "1803.08050",
    archivePrefix = "arXiv",
    primaryClass = "hep-th",
    doi = "10.1007/JHEP07(2018)124",
    journal = "JHEP",
    volume = "07",
    pages = "124",
    year = "2018",
    note = "[Erratum: JHEP 02, 197 (2019)]"
}

@article{Cipolloni:2024,
    author        = "Cipolloni, Giorgio",
    title         = "{Eigenvector Decorrelation for Random Matrices and Applications}",
    eprint        = "2410.10718",
    archivePrefix = "arXiv",
    primaryClass  = "math.PR",
    year          = "2024",
    url           = "https://arxiv.org/abs/2410.10718"
}

@article{Behrends:2019,
    author        = "Behrends, Jan and Bardarson, J{\'o}n H. and B{\'e}ri, Benjamin",
    title         = "{Tenfold way and many-body zero modes in the Sachdev-Ye-Kitaev model}",
    journal       = "Phys. Rev. B",
    volume        = "99",
    pages         = "195123",
    year          = "2019",
    doi           = "10.1103/PhysRevB.99.195123",
    eprint        = "1812.10853",
    archivePrefix = "arXiv",
    primaryClass  = "cond-mat.str-el"
}

@article{Saad:2018bqo,
    author        = "Saad, Phil and Shenker, Stephen H. and Stanford, Douglas",
    title         = "{A semiclassical ramp in SYK and in gravity}",
    eprint        = "1806.06840",
    archivePrefix = "arXiv",
    primaryClass  = "hep-th",
    year          = "2018",
    url           = "https://arxiv.org/abs/1806.06840"
}

@article{Baik:2005,
    author        = "Baik, Jinho and Ben Arous, G{\'e}rard and P{\'e}ch{\'e}, Sandrine",
    title         = "{Phase transition of the largest eigenvalue for nonnull complex sample covariance matrices}",
    journal       = "Ann. Probab.",
    volume        = "33",
    pages         = "1643--1697",
    year          = "2005",
    doi           = "10.1214/009117905000000233",
    eprint        = "math/0403022",
    archivePrefix = "arXiv"
}

@inproceedings{BennettBBV97,
    author        = "Bennett, Charles H. and Bernstein, Ethan and Brassard, Gilles and Vazirani, Umesh",
    title         = "{Strengths and Weaknesses of Quantum Computing}",
    journal       = "SIAM J. Comput.",
    volume        = "26",
    pages         = "1510--1523",
    year          = "1997",
    doi           = "10.1137/S0097539796300933",
    eprint        = "quant-ph/9701001",
    archivePrefix = "arXiv"
}

@article{Qi:2018bje,
    author        = "Qi, Xiao-Liang and Streicher, Alexandre",
    title         = "{Quantum Epidemiology: Operator Growth, Thermal Effects, and SYK}",
    journal       = "JHEP",
    volume        = "08",
    pages         = "012",
    year          = "2019",
    doi           = "10.1007/JHEP08(2019)012",
    eprint        = "1810.11958",
    archivePrefix = "arXiv",
    primaryClass  = "hep-th"
}

@article{Hayashi:1998,
    author        = "Hayashi, Masahito",
    title         = "{Asymptotic estimation theory for a finite-dimensional pure state model}",
    journal       = "J. Phys. A",
    volume        = "31",
    pages         = "4633--4655",
    year          = "1998",
    doi           = "10.1088/0305-4470/31/20/006"
}

@article{Keyl:2001,
    author        = "Keyl, Michael and Werner, Reinhard F.",
    title         = "{Optimal cloning of pure states, testing single clones}",
    journal       = "J. Math. Phys.",
    volume        = "40",
    pages         = "3283--3299",
    year          = "1999",
    doi           = "10.1063/1.532887",
    eprint        = "quant-ph/9807010",
    archivePrefix = "arXiv"
}

@article{Altland_1997,
   title={Nonstandard symmetry classes in mesoscopic normal-superconducting hybrid structures},
   volume={55},
   ISSN={1095-3795},
   url={http://dx.doi.org/10.1103/PhysRevB.55.1142},
   DOI={10.1103/physrevb.55.1142},
   number={2},
   journal={Physical Review B},
   publisher={American Physical Society (APS)},
   author={Altland, Alexander and Zirnbauer, Martin R.},
   year={1997},
   month=Jan, pages={1142–1161} }

@misc{leutheusser2024subregionsubalgebradualityemergencespace,
      title={Subregion-subalgebra duality: emergence of space and time in holography}, 
      author={Samuel Leutheusser and Hong Liu},
      year={2024},
      eprint={2212.13266},
      archivePrefix={arXiv},
      primaryClass={hep-th},
      url={https://arxiv.org/abs/2212.13266}, 
}

@article{PhysRevA.80.012304,
  title = {Exact and approximate unitary 2-designs and their application to fidelity estimation},
  author = {Dankert, Christoph and Cleve, Richard and Emerson, Joseph and Livine, Etera},
  journal = {Phys. Rev. A},
  volume = {80},
  issue = {1},
  pages = {012304},
  numpages = {6},
  year = {2009},
  month = {Jul},
  publisher = {American Physical Society},
  doi = {10.1103/PhysRevA.80.012304},
  url = {https://link.aps.org/doi/10.1103/PhysRevA.80.012304}
}

@article{Suzuki:2002zuq,
    author = "Suzuki, Masuo",
    title = "{Fractal decomposition of exponential operators with applications to many-body theories and Monte Carlo simulations}",
    doi = "10.1016/0375-9601(90)90962-N",
    journal = "Phys. Lett. A",
    volume = "146",
    pages = "319--323",
    month = "9",
    year = "2002"
}

@article{Suzuki:1991jtk,
    author = "Suzuki, Masuo",
    title = "{General theory of fractal path integrals with applications to many-body theories and statistical physics}",
    doi = "10.1063/1.529425",
    journal = "J. Math. Phys.",
    volume = "32",
    number = "2",
    pages = "400",
    year = "1991"
}

@article{Horodecki:1999,
    author        = "Horodecki, Micha{\l} and Horodecki, Pawe{\l} and Horodecki, Ryszard",
    title         = "{General teleportation channel, singlet fraction, and fidelity}",
    journal       = "Phys. Rev. A",
    volume        = "60",
    pages         = "1888--1898",
    year          = "1999",
    doi           = "10.1103/PhysRevA.60.1888",
    eprint        = "quant-ph/9807091",
    archivePrefix = "arXiv"
}

@book{Meckes:2019,
    author    = "Meckes, Elizabeth S.",
    title     = "{The Random Matrix Theory of the Classical Compact Groups}",
    publisher = "Cambridge University Press",
    address   = "Cambridge",
    year      = "2019",
    doi       = "10.1017/9781108303453"
}

@article{Ambainis:2002,
    author        = "Ambainis, Andris",
    title         = "{Quantum lower bounds by quantum arguments}",
    journal       = "J. Comput. Syst. Sci.",
    volume        = "64",
    number        = "4",
    pages         = "750--767",
    year          = "2002",
    doi           = "10.1006/jcss.2002.1826",
    eprint        = "quant-ph/0002066",
    archivePrefix = "arXiv"
}

@article{Rabinovici_2022,
   title={Krylov complexity from integrability to chaos},
   volume={2022},
   ISSN={1029-8479},
   url={http://dx.doi.org/10.1007/JHEP07(2022)151},
   DOI={10.1007/jhep07(2022)151},
   number={7},
   journal={Journal of High Energy Physics},
   publisher={Springer Science and Business Media LLC},
   author={Rabinovici, E. and Sánchez-Garrido, A. and Shir, R. and Sonner, J.},
   year={2022},
}

@misc{cryptoeprint:2023/1741,
      author = {Prabhanjan Ananth and Aditya Gulati and Fatih Kaleoglu and Yao-Ting Lin},
      title = {Pseudorandom Isometries},
      howpublished = {Cryptology {ePrint} Archive, Paper 2023/1741},
      year = {2023},
      url = {https://eprint.iacr.org/2023/1741}
}

@article{Czech:2012bh,
  author        = {Czech, Bartlomiej and Karczmarek, Joanna L. and Nogueira, Fernando and Van Raamsdonk, Mark},
  title         = {{The Gravity Dual of a Density Matrix}},
  eprint        = {1204.1330},
  archivePrefix = {arXiv},
  primaryClass  = {hep-th},
  doi           = {10.1088/0264-9381/29/15/155009},
  journal       = {Class. Quant. Grav.},
  volume        = {29},
  pages         = {155009},
  year          = {2012}
}

@article{Dong:2016eik,
  author        = {Dong, Xi and Harlow, Daniel and Wall, Aron C.},
  title         = {{Reconstruction of Bulk Operators within the Entanglement Wedge in Gauge-Gravity Duality}},
  eprint        = {1601.05416},
  archivePrefix = {arXiv},
  primaryClass  = {hep-th},
  doi           = {10.1103/PhysRevLett.117.021601},
  journal       = {Phys. Rev. Lett.},
  volume        = {117},
  number        = {2},
  pages         = {021601},
  year          = {2016}
}

@article{Hawking:1982dh,
  author  = {Hawking, S. W. and Page, Don N.},
  title   = {{Thermodynamics of Black Holes in anti-De Sitter Space}},
  doi     = {10.1007/BF01208266},
  journal = {Commun. Math. Phys.},
  volume  = {87},
  pages   = {577},
  year    = {1983}
}

@article{Witten:1998zw,
  author        = {Witten, Edward},
  title         = {{Anti-de Sitter Space, Thermal Phase Transition, and Confinement in Gauge Theories}},
  eprint        = {hep-th/9803131},
  archivePrefix = {arXiv},
  doi           = {10.4310/ATMP.1998.v2.n3.a3},
  journal       = {Adv. Theor. Math. Phys.},
  volume        = {2},
  pages         = {505--532},
  year          = {1998}
}

@article{Massar:1995,
  author  = {Massar, Serge and Popescu, Sandu},
  title   = {{Optimal Extraction of Information from Finite Quantum Ensembles}},
  doi     = {10.1103/PhysRevLett.74.1259},
  journal = {Phys. Rev. Lett.},
  volume  = {74},
  pages   = {1259--1263},
  year    = {1995}
}

@article{Acin:2001,
  author        = {Ac{\'i}n, A. and Jan{\'e}, E. and Vidal, G.},
  title         = {{Optimal Estimation of Quantum Dynamics}},
  eprint        = {quant-ph/0012015},
  archivePrefix = {arXiv},
  doi           = {10.1103/PhysRevA.64.050302},
  journal       = {Phys. Rev. A},
  volume        = {64},
  pages         = {050302},
  year          = {2001}
}

@article{Bisio:2009,
  author        = {Bisio, Alessandro and Chiribella, Giulio and D'Ariano, Giacomo Mauro and Facchini, Stefano and Perinotti, Paolo},
  title         = {{Optimal Quantum Learning of a Unitary Transformation}},
  eprint        = {0903.0543},
  archivePrefix = {arXiv},
  primaryClass  = {quant-ph},
  doi           = {10.1103/PhysRevA.81.032324},
  journal       = {Phys. Rev. A},
  volume        = {81},
  pages         = {032324},
  year          = {2010}
}

@article{Maldacena_1999,
   title={The Large-N Limit of Superconformal Field Theories and Supergravity},
   volume={38},
   ISSN={1572-9575},
   url={http://dx.doi.org/10.1023/A:1026654312961},
   DOI={10.1023/a:1026654312961},
   number={4},
   journal={International Journal of Theoretical Physics},
   publisher={Springer Science and Business Media LLC},
   author={Maldacena, Juan},
   year={1999},
   month=Apr, pages={1113–1133} }

@article{Almheiri:2014lwa,
  author        = {Almheiri, Ahmed and Dong, Xi and Harlow, Daniel},
  title         = {{Bulk Locality and Quantum Error Correction in AdS/CFT}},
  eprint        = {1411.7041},
  archivePrefix = {arXiv},
  primaryClass  = {hep-th},
  doi           = {10.1007/JHEP04(2015)163},
  journal       = {JHEP},
  volume        = {04},
  pages         = {163},
  year          = {2015}
}

@article{Low:2009,
  author        = {Low, Richard A.},
  title         = {{Pseudo-randomness and Learning in Quantum Computation}},
  eprint        = {1006.5227},
  archivePrefix = {arXiv},
  primaryClass  = {quant-ph},
  year          = {2010}
}

@article{Roberts:2018mnp,
  author        = {Roberts, Daniel A. and Stanford, Douglas and Streicher, Alexandre},
  title         = {{Operator growth in the SYK model}},
  eprint        = {1802.02633},
  archivePrefix = {arXiv},
  primaryClass  = {hep-th},
  doi           = {10.1007/JHEP06(2018)122},
  journal       = {JHEP},
  volume        = {06},
  pages         = {122},
  year          = {2018}
}

@article{Parker:2019yvk,
  author        = {Parker, Daniel E. and Cao, Xiangyu and Avdoshkin, Alexander and Scaffidi, Thomas and Altman, Ehud},
  title         = {{A Universal Operator Growth Hypothesis}},
  eprint        = {1812.08657},
  archivePrefix = {arXiv},
  primaryClass  = {cond-mat.stat-mech},
  doi           = {10.1103/PhysRevX.9.041017},
  journal       = {Phys. Rev. X},
  volume        = {9},
  number        = {4},
  pages         = {041017},
  year          = {2019}
}

\end{document}